\documentclass[10pt,a4paper]{amsart}

\usepackage{amssymb} 
\usepackage[all]{xy} 
\usepackage{graphicx} 
\usepackage{eucal} 
\usepackage{hyperref} 
\usepackage{color} 

\numberwithin{equation}{section}

\newtheorem{definition}{Definition}[section]
\newtheorem{lemma}[definition]{Lemma}
\newtheorem{theorem}[definition]{Theorem}
\newtheorem{prop}[definition]{Proposition}
\newtheorem{corollary}[definition]{Corollary}
\newtheorem{remark}[definition]{Remark}

\newtheorem{state}[definition]{Statement}

\newcommand{\ds}{\displaystyle}

\def\beq{\begin{equation}}
\def\eeq{\end{equation}}
\def\bea{\begin{eqnarray}}
\def\eea{\end{eqnarray}}
\def\beann{\begin{eqnarray*}}
\def\eeann{\end{eqnarray*}}
\def\beasn{\begin{sneqnarray}}
\def\eeasn{\end{sneqnarray}}
\def\ben{\begin{enumerate}}
\def\een{\end{enumerate}}
\def\bit{\begin{itemize}}
\def\eit{\end{itemize}}

\def\r{\ensuremath{\mathbb{R}}}
\def\rk{{\mathbb R}^{k}}

\def\rk{{\mathbb R}^{k}}

\def\g^*{\mathfrak{g}^*}

\def\derpar#1#2{\displaystyle\frac{\partial{#1}}{\partial{#2}}}

\def\qed{\ifvmode\Realemovelastskip\fi
{\unskip\nobreak\hfil\penalty50\hbox{}\nobreak\hfil \hbox{\vrule
height1.2ex width1.2ex}\parfillskip=0pt \finalhyphendemerits=0
\par\smallskip}}

\def\qed{\ifvmode\removelastskip\fi
{\unskip\nobreak\hfil\penalty50\hbox{}\nobreak\hfil \hbox{\vrule
height1.2ex width1.2ex}\parfillskip=0pt \finalhyphendemerits=0
\par\smallskip}}

\newcounter{epigrafe}[subsection]

\newcounter{subepigrafe}[epigrafe]

\title{Reduction of polysymplectic manifolds}

\author[J.C. Marrero]{Juan Carlos Marrero}
\address{Juan Carlos Marrero:
ULL-CSIC Geometr\'{\i}a Diferencial y Mec\'anica Geom\'etrica,
    Departamento de Matem\'aticas, Estad\'istica e Investigaci\'on Operativa, 
Universidad de la Laguna, Spain}
    \email{jcmarrer@ull.edu.es}

\author[N. Rom\'an-Roy]{Narciso Rom\'an-Roy}
\address{Narciso Rom\'an-Roy:
    Departamento de Matem\'atica Aplicada IV.
    Edificio C-3, Campus Norte UPC.
Universitat Polit\`ecnica de Catalunya - Barcelona Tech.
    C/ Jordi Girona 1. 08034 Barcelona, Spain}
    \email{nrr@ma4.upc.edu}

\author[M. Salgado]{Modesto Salgado}
\address{Modesto Salgado:
    Departamento de Xeometr\'{\i}a e Topolox\'{\i}a,
    Facultad de Matem\'{a}ticas, Campus Sur USC,
    C/ Lope G\'{o}mez de Marzoa, s/n. 15782 Santiago de Compostela, Spain}
    \email{modesto.salgado@usc.es}
    
\author[S. Vilari\~no]{Silvia Vilari\~no}
\address{Silvia Vilari\~no:
    Centro Universitario de La Defensa  de Zaragoza $\&$ I.U.M.A.,
    Academia General Militar,
    Carretera de Huesca s/n,
    50090-Zaragoza, Spain}
    \email{silviavf@unizar.es}


\begin{document}

\begin{abstract}
    The aim of this paper is to generalize the classical Marsden-Weinstein reduction procedure
    for symplectic manifolds to polysymplectic manifolds in order to obtain quotient
    manifolds which inherit the polysymplectic structure.
    This generalization allows us to reduce polysymplectic Hamiltonian systems
    with symmetries, such as those appearing in certain kinds of classical field theories.
    As an application of this technique,
    an analogous to the Kirillov-Kostant-Souriau theorem for polysymplectic manifolds is obtained
    and some other mathematical examples are also analyzed.

Our procedure corrects some mistakes and inaccuracies in previous papers
\cite{Gunther-1987,MRS-2004} on this subject.
\end{abstract}


\maketitle

\begin{center}
({\sl J. Phys. A: Math. Theor.} {\bf 48} (2015) 055206 (43pp) doi:10.1088/1751-8113/48/5/055206)
\end{center}

 \tableofcontents



\section{Introduction}

The problem of reduction of systems with symmetry
has attracted the interest of theoretical physicists
and mathematicians, who have sought to reduce the number of
equations describing the behavior of the system
by finding first integrals or conservation laws.
The use of geometrical methods has proved to be a powerful tool in the study of this topic,
and was introduced by Marsden and Weinstein in their pioneering work
of reduction of autonomous Hamiltonian systems under the action of
a Lie group of symmetries, with regular values of their momentum maps \cite{MW-1974}
(see also \cite{MW-2001} for a review of symplectic reduction).
In this case, the reduced phase space so-obtained is a symplectic manifold and inherits a Hamiltonian
dynamics from the initial system.

The Marsden-Weinstein technique was subsequently applied and generalized to many different situations;
for instance, the reduction of Hamiltonian systems with singular values of the momentum map
has been studied in several papers such as
\cite{SL-91} for the autonomous case, and  \cite{LS-93} for the non-autonomous.
In both cases, a stratified symplectic manifold is obtained as a quotient manifold which,
in the second situation, is also endowed with a cosymplectic structure.
Furthermore, with certain additional conditions, the reduced phase space inherits
a non-degenerate Poisson structure \cite{ACG-91} (see also other references quoted therein).
The reduction of time-dependent regular Hamiltonian systems (with regular values)
is developed in the framework of cosymplectic manifolds in \cite{Al-89},
obtaining a reduced phase space which is a cosymplectic manifold.
The study of autonomous systems coming from certain kinds of singular Lagrangians
can be found in \cite{CCCI-86}, where the conditions for
the reduced phase space to inherit an almost-tangent structure are given.
Some of the results here obtained are generalized
to the case of non-autonomous singular Lagrangian systems in \cite{IM-92}.
Another approach to this question is adopted in \cite{LMR-92}, where the authors give conditions for
the existence of a regular Lagrangian function in the reduced phase space,
which allows them to construct the reduced cosymplectic or contact structure
(and hence the reduced Hamiltonian function) from it.
Finally, a general study on reduction of presymplectic Hamiltonian systems with symmetry
is conducted in \cite{EMR-99}.

There are further cases in reduction theory; for instance,
the theory of reduction of Poisson manifolds is treated in works such as \cite{LM-95} and \cite{MR-86}.
Reduction of cotangent bundles of Lie groups is considered in \cite{MRW-84}.
As regards the subject of Lagrangian reduction,
some works, such as \cite{MS-93}, consider the problem from the point of view of
reducing variational principles (instead of reducing the almost tangent structure, as is
the case made in some of the above mentioned references), as well as
other approaches to the so-called Euler-Poincar\'e reduction \cite{CGR-2001,CRS-2000}
and Routh reduction for regular and singular Lagrangians \cite{CL-2010,LCV-2010}.
The study of reduction of non-holonomic systems
can be found, for instance, in \cite{BS-93}, \cite{BKMM-1996}, \cite{CLMM-98} and \cite{Ma-95}.
Finally, in \cite{BC-99} a presentation of optimal control systems on coadjoint orbits related to reduction
problems and integrability is provided, although it is in previous papers such as \cite{Su-95} and
\cite{Sh-87}, where an initial analysis of the problem of symmetries of optimal control systems
is carried out.
A more general treatment of the reduction problem of these kinds of systems
using the reduction theory for presymplectic systems is given in \cite{EMMR-2003}.
A different point of view on this topic using Dirac structures and implicit Hamiltonian systems
is adopted in \cite{Bl-2002} and \cite{BVS-2000}; while a further approach can be found in \cite{MCL-01}.
(Of course, this list of references is far from being complete).

With regard to the problem of reduction by symmetries of classical field theories,
only partial results have been achieved in the context of the Lagrangian and Poisson
reduction, leading to the analogous of the Lie-Poisson equation in classical mechanics
\cite{CM-2003}, the Euler-Poincar\'e reduction in principal fiber bundles \cite{CGR-2007,CR-2003}
and for discrete field theories \cite{Va-2007},
and other particular situations in multisymplectic field theories.
Nevertheless, although studies on symmetries and conservation laws in field theories
have already been carried out (see, for instance,
\cite{EMR-99b,Gimmsy,LMS-2004,MRSV-2010,RSV-2007}
and the references quoted therein),
a complete generalization of the Marsden-Weinstein reduction theorem
to the case of classical field theory has yet to be obtained.

The main objective of this paper is  to perform this generalization for one of the
simplest geometric formalisms of classical field theories:
the so-called {\sl $k$-symplectic formalism}  \cite{Gunther-1987} (on its Hamiltonian formulation),
and considering only the regular case.
This $k$-symplectic formalism
(also called {\sl polysymplectic formalism}\/) is the generalization
to field theories of the standard symplectic formalism in
autonomous mechanics, and is used to give a geometric description of certain kinds
of field theories: in a local description, those whose Lagrangian and Hamiltonian functions
do not depend on the coordinates in the basis (in many of these theories,
the space-time coordinates). The foundations of the
$k$-symplectic formalism are  the {\sl $k$-symplectic manifolds} \cite{aw-1992,aw-1994,AG-2000,mt1}.

An initial approach to reduction in this context was made in the seminal work of
G\"{u}nther \cite{Gunther-1987},
where the author attempts to apply the Marsden-Weinstein reduction theory for symplectic manifolds
to the polysymplectic case. Nevertheless,
in this paper (in which the author wishes to generalize some technical properties of
the orthogonal symplectic complement to the analogous polysymplectic situation)
the proof of one of the fundamental results fails to hold true.
A more recent attempt was made in \cite{MRS-2004} for reduction of
$k$-symplectic structures,
but this article contains similar inaccuracies that invalidate the proof of
the theorem of reduction of the polysymplectic structure proposed there.
On the other hand, a further analogous erroneous attempt to extend the
Marsden-Weinstein reduction theorem to multisymplectic manifolds was made in \cite{Hrabak}.
A promising way to address this problem has
been initiated very recently by Bursztyn et al \cite{BuCaIg}. The key point in this approach is
to use the notion of a multiplicative form in a Lie groupoid (see \cite{BuCa, BuCaOr}).
Another approach using a different and appropriate notion of a multi-momentum map
was proposed by Madsen and Swann \cite{MaSw1, MaSw2} (see also \cite{Sw}). The theory
is applied to closed forms of arbitrary degree. Existence and uniqueness of multi-momentum maps
was discussed and applications to the reduction of several types of ``closed geometries of higher order''
are given.

In this paper, we seek to correct the previous inaccuracies, although as
we will see, the generalization of the Marsden-Weinstein theorem to the polysymplectic context
(for regular values of the corresponding momentum maps)
is not straightforward and some additional technical conditions must be added to
the usual hypothesis. We also study how a polysymplectic structure
can be defined in the quotient space, and then, when starting
from a Hamiltonian polysymplectic system, how to reduce it.

The organization of the paper is as follows:
Section \ref{poly_manifolds} provides
a brief review on polysymplectic manifolds  (in appendix \ref{examplepoly}
we present some typical examples of these structures).
In particular, we review G\"{u}nther's reduction method
and give a counterexample showing that this procedure is not correct.
The main results of the paper are presented in Section \ref{main_reduction},
where we study the reduction procedure for polysymplectic structures in general,
first considering the reduction by a submanifold in general, and then
stating the Marsden-Weinstein reduction theorem for this case.
As an application, some typical examples are analyzed; namely, the reduction of the product of symplectic manifolds,
the reduction of cotangent bundles of $k^1$-covelocities and
the Kirillov-Kostant-Souriau theorem for polysymplectic manifolds.
In Section \ref{Ham-red}, the above results are applied and completed in order to reduce
polysymplectic Hamiltonian systems, and the procedure is applied to
certain kinds of Hamiltonian polysymplectic systems defined in
cotangent bundles of $k^1$-covelocities,
as well as to the problem of harmonic maps, as a particular example.

 Throughout this work, manifolds are real, paracompact,
 connected and $C^\infty$, maps are $C^\infty$, and sum over crossed repeated
 indices is understood. $G$ denotes a Lie group and ${\mathfrak g}$ its Lie algebra.

\section{Comments on G\"{u}nther's polysymplectic reduction: A counterexample.}
\label{poly_manifolds}

In \cite{Gunther-1987}, G\"{u}nther extends the Marsden-Weinstein reduction \cite{MW-1974}
to the polysymplectic setting. However, as commented in the introduction to the present paper,
the description given by G\"{u}nther contains some mistakes.
In this section we discuss this fact and present a simple counterexample to G\"{u}nther's results;
in particular, we see that Lemma 7.5 and Theorem 7.6 in \cite{Gunther-1987} are incorrect.
First, we recall the notions of a polysymplectic manifold,
a polysymplectic action and momentum map, and then in section \ref{Gu-error}
we discuss G\"{u}nther's results on reduction.

\subsection{Polysymplectic manifolds, actions and momentum maps.}

In this section we review the concept of a polysymplectic structure introduced by G\"{u}nther
 in \cite{Gunther-1987} and some necessary notions for the reduction procedure described by this author
(for further details see \cite{Gunther-1987} and also \cite{MRS-2004}).

    \begin{definition}
        Let $M$ be a differentiable manifold of dimension $n$.
        A {\rm $k$-polysymplectic structure} in  $M$ is a closed nondegenerated $\rk$-valued $2$-form
        $$\bar{\omega}=\displaystyle\sum_{A=1}^k \omega^A\otimes r_A \ ,$$
 where $\{r_1,\ldots, r_k\}$ denotes the canonical basis of $\rk$.
         The pair $(M,\bar\omega)$ is called a $k$-polysymplectic manifold or simply a polysymplectic manifold.
    \end{definition}

   Some typical examples of polysymplectic manifolds are analyzed in Appendix A.

   Note that $M$ has a $k$-polysymplectic structure $\bar{\omega}$ if and only if there exists a family of $k$ closed $2$-forms $(\omega^1,\ldots,\omega^k)$ such that
                 \begin{equation}\label{poly-cond}
                     \displaystyle\bigcap_{A=1}^k \ker\,\omega^A=0\, .
                 \end{equation}

    Throughout this paper we use this characterization of a polysymplectic structure. Thus,
    a family of $k$ closed $2$-forms $(\omega^1,\ldots, \omega^k)$ such that (\ref{poly-cond})
    holds is called a $k$-\textit{polysymplectic structure} or
    simply a {\it polysymplectic structure}.

    \begin{remark}
    {\rm
        The definition of a polysymplectic manifold is the differentiable version
of the notion of a polysymplectic vector space: \emph{a polysymplectic structure on a
vector space $\mathcal{V}$ is a family of $k$ skew-symmetric bilinear maps $\omega^1,\ldots, \omega^k$
such that $\ker\, \omega^1\cap \ldots \cap \ker\, \omega^k=\{0\}$.}
    }
    \end{remark}

    \begin{definition}
            An action $\Phi\colon G\times M\to M$ of a Lie group $G$ on a polysymplectic manifold
            $(M,\omega^1,\ldots, \omega^k)$, is said to be a {\rm polysymplectic action}
            if for each $g\in G$, the diffeomorphism
$$\begin{array}{ccccc}
\Phi_g &\colon& M & \to & M \\ & & x & \mapsto & \Phi(g,x)
\end{array}$$
 is polysymplectic; that is,  for $A=1,\ldots, k$,
                \[
                    \Phi_g^*\omega^A=\omega^A\,.
                \]
        \end{definition}

As in the symplectic case, we can introduce the notion of a
momentum map for polysymplectic actions in a natural way:

        \begin{definition}\label{momemtum}
            Let $(M,\omega^1,\ldots, \omega^k)$ be a polysymplectic manifold and
$\Phi\colon G\times M\to M$ a polysymplectic action.
            A mapping
                \[
                    J\equiv (J^1,\ldots, J^k)\colon M\to \mathfrak{g}^*\times\stackrel{k}{\ldots}\times\mathfrak{g}^*
                \]
            is said to be a {\rm momentum mapping} for the action $\Phi$
            if for each $\xi\in\mathfrak{g}$,
                \[
                    i_{\xi_M}\omega^A=d\hat{J}^A_\xi\ ,
                \]
            where $\hat{J}^A_\xi\colon M\to \mathbb{R}$ is the map defined by
                \[
                    \hat{J}^A_\xi(x)=J^A(x)(\xi)\,,\; x\in M
                \]
            and $\xi_M$ is the infinitesimal generator of the action
            corresponding to $\xi$.
        \end{definition}

        \begin{remark}{\rm
            In the particular case $k=1$, the above definition reduces to the definition of
the momentum mapping for a symplectic action.
            (See \cite{AM-1978}).
       } \end{remark}

        If $G$ is a  Lie group, we may  define an action of $G$ over
$\mathfrak{g}^*\times\stackrel{k}{\ldots}\times\mathfrak{g}^*$ by
            \begin{equation}\label{coad^k}
                \begin{array}{lccl}
                    Coad^k\colon & G\times \mathfrak{g}^*\times\stackrel{k}{\ldots}\times\mathfrak{g}^*
                    &\to & \mathfrak{g}^*\times\stackrel{k}{\ldots}\times\mathfrak{g}^*\\\noalign{\medskip}
                            & (g,\mu_1,\ldots, \mu_k) & \mapsto & Coad^k(g,\mu_1,\ldots, \mu_k)=
\left( Coad(g,\mu_1),\ldots, Coad(g,\mu_k)\right) \ ,
                \end{array}
            \end{equation}
where $Coad$ denotes the usual coadjoint action
$$
\begin{array}{ccccc}
Coad & \colon & G\times{\mathfrak g}^* & \to {\mathfrak g}^* \\
& & (g,\mu) & \mapsto & \mu\circ Ad_{g^{-1}}
\end{array}
$$
        $Coad^k$ is called the \textit{$k$-coadjoint action} (see Appendix A).
        \begin{definition}
            A momentum mapping
            $J\equiv (J^1,\ldots, J^k)\colon M\to \mathfrak{g}^*\times\stackrel{k}{\ldots}\times\mathfrak{g}^*$
             for the action $\Phi$ is said to be $Coad^k$-equivariant if, for every $g\in G$ and $x\in M$,
                \begin{equation}\label{equivariant}
                    J(\Phi_g(x))=Coad^k_g(J(x))\, ;
                \end{equation}
            that is, the following diagram is commutative
                \[
 \xymatrix{M \ar[d]_-{\Phi_g}\ar[r]^-{J}&
\mathfrak{g}^*\times\stackrel{k}{\ldots}\times\mathfrak{g}^*\ar[d]^-{Coad^k_g} \\
                    M \ar[r]^-{J}& \mathfrak{g}^*\times\stackrel{k}{\ldots}\times\mathfrak{g}^* }
                \]
        \end{definition}

        \begin{remark}{\rm
            \begin{enumerate}
            \item
            Observe that, for every $g\in G$ and $x\in M$, the condition (\ref{equivariant}) is
            equivalent to
                \[
                    J^A(\Phi_g(x))=Coad_g(J^A(x))\ , \quad \mbox{\rm for every $A=1,\ldots,k$} \ .
                \]
            \item If $J$ is $Coad^k$-equivariant then $T_mJ(\xi_M(m))=
\xi_{\mathfrak{g}^*\times\stackrel{k}{\ldots}\times\mathfrak{g}^*}(J(m))$, for $m\in M$ and
$\xi\in \mathfrak{g}$, where $\xi_{\mathfrak{g}^*\times\stackrel{k}{\ldots}\times\mathfrak{g}^*}$
 is the infinitesimal generator of $Coad^k$ associated with $\xi$.
            \end{enumerate}
        }\end{remark}

\begin{definition}
A polysymplectic manifold endowed with a polysymplectic action of a Lie group
and a $Coad^k$-equivariant momentum map, $(M;\omega^1,\ldots,\omega^k;\Phi; J)$,
is said to be a {\rm polysymplectic Hamiltonian $G$-space}.
\end{definition}

In this setting we can prove a result which generalizes  Lemma 4.3.2 in \cite{AM-1978}.
First we need to introduce the following concept:
let $(\mathcal{V},\omega^1,\ldots, \omega^k)$ be a polysymplectic vector space and $W$ be a subspace.
The \textit{polysymplectic orthogonal complement} of $W $
 is the linear subspace of $\mathcal{V}$ defined by
$$
        W^{\bot,k}=\{ v\in \mathcal{V}\, |\, \omega^1(v,w)=\ldots =
\omega^k(v,w)=0,\, \makebox{for every } w\in W\} = \bigcap_{A=1}^k W^{\bot, \omega^A}\,.
$$
(A complete description of the $k$-th orthogonal complement and its properties
can be found in \cite{LV-2012}).
Then:

\begin{lemma}\label{amlemma}
    Let $\Phi\colon G\times M\to M$ be a polysymplectic action with momentum mapping
$J\colon M\to \mathfrak{g}^*\times\stackrel{k}{\ldots}\times \mathfrak{g}^*$, and let
$\mu\in \mathfrak{g}^*\times\stackrel{k}{\ldots}\times \mathfrak{g}^*$
be a regular value of $J$. If $m\in J^{-1}(\mu)$ and $G_\mu$ is the isotropy group of $\mu$
 under the $k$-coadjoint action, we have:
    \begin{enumerate}
        \item $T_m(G_\mu\cdot m) = T_m(G\cdot m)\cap T_m(J^{-1}(\mu))$ and
        \item $T_m(J^{-1}(\mu))= T^{\bot,k}_m(G\cdot m)$, where $^{\bot, k}$ denotes the polysymplectic orthogonal complement.
    \end{enumerate}
\end{lemma}
\begin{proof}
For $(1)$, observe that $v\in T_m(G\cdot m)$ if and only if there exists $\xi \in \mathfrak{g}$
such that $v=\xi_M(m)$. Then, to check $(1)$ is equivalent to proving that
$\xi_M(m)\in T_m(G_\mu\cdot m)$ if and only if $\xi_M(m)\in T_m(J^{-1}(\mu))$,
or equivalently $\xi\in \mathfrak{g}_{\mu}$ if and only if $\xi_M(m)\in T_m(J^{-1}(\mu))$.

Now, note that $\xi_M(m)\in T_m(J^{-1}(\mu))$ if and only if $T_mJ(\xi_M(m))=0$,
 that is $\xi_{\mathfrak{g}^*\times\stackrel{k}{\ldots}\times \mathfrak{g}^*}(J(m))=0$.
Since $m\in J^{-1}(\mu)$, we have that
$\xi_{\mathfrak{g}^*\times\stackrel{k}{\ldots}\times \mathfrak{g}^*}(\mu)=
\xi_{\mathfrak{g}^*\times\stackrel{k}{\ldots}\times \mathfrak{g}^*}(J(m))=0$
and then $\xi\in \mathfrak{g}_\mu$. Therefore $(1)$ holds.

For the item $(2)$, we have
\begin{align*}
    X\in T_m^{\bot,k}(G\cdot m) \Leftrightarrow\, &
\omega^A(m)(X,\xi_M(m))=0,\, \forall \xi\in \mathfrak{g} \makebox{ and }\forall A=1,\ldots, k\\ \Leftrightarrow \,
&
    d\hat{J}^A_\xi (m)(X) = 0 ,\, \forall \xi\in \mathfrak{g} \makebox{ and }\forall A=1,\ldots, k\\ \Leftrightarrow \,
&
    T_mJ^A(X)=0,\; \forall A=1,\ldots, k \\ \Leftrightarrow \, & X\in T_m(J^{-1}(\mu))\,.
\end{align*}
\end{proof}

\subsection{G\"{u}nther's reduction: a counterexample.}\label{Gu-error}

The idea of the reduction of polysymplectic manifolds is to generalize the
Marsden-Weinstein reduction procedure for symplectic manifolds to polysymplectic manifolds
in order to obtain quotient manifolds which inherit the polysymplectic structure.

A first but incomplete attempt at reduction in this setting was made in \cite{Gunther-1987}
 (see also \cite{MRS-2004}). In this direction, the main result in G\"{u}nther's paper is the following statement:

\begin{state}\label{guth}
    Let $\Phi\colon G\times M\to M$ be a polysymplectic action with momentum map
$J\colon M\to \mathfrak{g}^*\times\stackrel{k}{\ldots}\times \mathfrak{g}^*$, and let
$\mu\in \mathfrak{g}^*\times\stackrel{k}{\ldots}\times \mathfrak{g}^*$
a regular value of $J$. Then there exists uniquely a polysymplectic form
$\bar{\omega}_{\mu}$ on $M_\mu=J^{-1}(\mu)/G_\mu$ with
$\pi_\mu^*\bar{\omega}_\mu= i_\mu^*\bar{\omega}$, where
$\pi_\mu\colon J^{-1}(\mu)\to M_\mu$ is the canonical projection and
$i_\mu\colon J^{-1}(\mu)\to M$ is the canonical inclusion.
\end{state}

The proof of this statement is based on the following result (Lemma 7.5 in \cite{Gunther-1987}).

\begin{state}\label{gunlemma}
    Under the same conditions as in the statement \ref{guth}, if $m\in J^{-1}(\mu)$ the following relations hold:
    \begin{enumerate}
        \item $T_m(J^{-1}(\mu))=T_m^{\bot,k}(G\cdot m),$
        \item $T_m(G_\mu\cdot m)=T_m^{\bot,k}(G\cdot m) \cap T_m^{\bot,k} (J^{-1}(\mu))\,.$
    \end{enumerate}
\end{state}

Let us observe that the above statement is true for symplectic manifolds
(and in this case it coincides with Lemma \ref{amlemma}), but in general
it is not true for polysymplectic manifolds. The key point is that if $W$
is a subspace of a polysymplectic vector space $(\mathcal{V}, \omega^1,\ldots, \omega^k)$
 then it is not true, in general, that $(W^{\bot,k})^{\bot,k}= W$,
and in the above lemma G\"{u}nther assumes that the identity $(W^{\bot,k})^{\bot,k}= W$ holds.
Next,  we present a simple counterexample of the above results.

Let $(N,\omega)$ be a symplectic manifold, then $M =N\times N$
has a polysymplectic structure given by $\omega^A=pr_A^*\omega,\, A=1,2$,  $pr_1$
and $pr_2$ being the canonical projections.

Let $\phi\colon G\times N\to N$ be a free and proper symplectic action with
equivariant momentum mapping $\tilde{J}\colon N\to \mathfrak{g}^*$.
Then we can define a free and proper polysymplectic action by
\[
    \begin{array}{lccl}
        \Phi\colon & G\times (N\times N) & \to & N\times N\\\noalign{\medskip}
         & (g, (x,y)) & \mapsto & (\phi_g(x), \phi_g(y))
    \end{array}
\]
and a $Coad^2$-equivariant  momentum mapping for $\Phi$ given by
\[
    \begin{array}{lccl}
       	J\colon & M=N\times N &\to &\mathfrak{g}^*\times \mathfrak{g}^*\\\noalign{\medskip}
         & (x,y) & \to & (\tilde{J}(x), \tilde{J}(y))
    \end{array}
\]

Let $\mu=(\mu_1, \mu_2)\in \mathfrak{g}^*\times \mathfrak{g}^*$.
Since the action $\phi$ is free and proper, $\mu_1$ and $\mu_2$ are regular values of $\tilde{J}$,
and then $\mu$ is a regular value of $J$. Therefore, $G_\mu$
acts free and properly on $J^{-1}(\mu)$ and this
implies that $J^{-1}(\mu)/G_\mu$ is a smooth quotient manifold.

Next, we see that, for this example, item $(2)$ in the statement \ref{gunlemma}
does not hold. In fact, we know that
\beann
    T_{(x_1,x_2)}J^{-1}(\mu) &=& \{(v_1,v_2)\in T_{x_1}N\times T_{x_2}N\,\vert\, T_{x_1}\tilde{J}(v_1)=0, \,
T_{x_2}\tilde{J}(v_2)=0\} \\ 
 &=& T_{x_1}(\tilde{J}^{-1}(\mu_1))\times T_{x_2}(\tilde{J}^{-1}(\mu_2)),\\
    T_{(x_1,x_2)}(G\cdot (x_1,x_2)) &=& \{ (\xi_N(x_1), \xi_N(x_2))\, \vert\, \xi\in \mathfrak{g}\}
\eeann
and, as a consequence of item $(2)$ in Lemma \ref{amlemma}, we have that
$$
    T_{(x_1,x_2)}^{\bot,2}(G\cdot (x_1,x_2)) =
 T_{x_1}(\tilde{J}^{-1}(\mu_1))\times T_{x_2}(\tilde{J}^{-1}(\mu_2))\,.
$$
On the other hand, using again Lemma \ref{amlemma}, we know that
\begin{align}\label{cond1ex}
    T_{(x_1,x_2)}\big(G_\mu\cdot (x_1,x_2)\big) =& T_{(x_1,x_2)}\big(G\cdot (x_1,x_2)\big) \cap
T_{(x_1,x_2)}(\widetilde{J}^{-1}(\mu)) \\\nonumber
    =&\{(\xi_N(x_1), \xi_N(x_2))\, \vert\, \xi\in \mathfrak{g}_{\mu_1}\cap \mathfrak{g}_{\mu_2}\}\,.
\end{align}

Finally,
\begin{align}
\nonumber
     T_{(x_1,x_2)}J^{-1}(\mu) \cap T_{(x_1,x_2)}^{\bot, 2}J^{-1}(\mu)
    =& \Big(T_{x_1}(\tilde{J}^{-1}(\mu_1))\times T_{x_2}(\tilde{J}^{-1}(\mu_2))\Big)
 \cap \Big(T_{x_1}(\tilde{J}^{-1}(\mu_1))\times T_{x_2}(\tilde{J}^{-1}(\mu_2))\Big)^{\bot,2}\nonumber \\
    =& \Big(T_{x_1}(\tilde{J}^{-1}(\mu_1))\times T_{x_2}(\tilde{J}^{-1}(\mu_2))\Big)
\cap \Big(T_{x_1}^{\bot}(\tilde{J}^{-1}(\mu_1))\times T_{x_2}^\bot (\tilde{J}^{-1}(\mu_2))\Big)
    \nonumber \\
    =&\Big(T_{x_1}(\tilde{J}^{-1}(\mu_1))\cap T_{x_1}^{\bot}(\tilde{J}^{-1}(\mu_1))\Big)\times
 \Big(T_{x_2}(\tilde{J}^{-1}(\mu_2))\cap T_{x_2}^\bot (\tilde{J}^{-1}(\mu_2))\Big)
    \nonumber \\
    =& T_{x_1}(G_{\mu_1}\cdot x_1)\times T_{x_2}(G_{\mu_2}\cdot x_2) =
 \{(\xi_N(x_1),\eta_N(x_2))\, \vert\, \xi\in \mathfrak{g}_{\mu_1},\eta\in \mathfrak{g}_{\mu_2}\}\,.
\label{condfin}
\end{align}

\begin{remark}
{\rm
In (\ref{condfin}) the symbol $^\bot$ denotes the symplectic orthogonal of a subspace.
Moreover, we use the following result: \emph{ If $(V,\omega)$ is a symplectic vector space,
 and $W,W'$ are two subspaces of the vector space $V$, then
 $(W\times W')^{\bot, 2} = W^\bot\times (W')^\bot$. }
%
}\end{remark}

Using (\ref{cond1ex}) and (\ref{condfin}), it follows that
$T_{(x_1,x_2)}\big(G_\mu\cdot (x_1,x_2)\big) \subset T_{(x_1,x_2)}J^{-1}(\mu) \cap
 T_{(x_1,x_2)}^{\bot, 2}J^{-1}(\mu)$, but in general these two spaces are different.
Therefore, item $(2)$ in the statement \ref{gunlemma} is not always right.
This implies that the quotient space $M_\mu=J^{-1}(\mu)/G_\mu$
 is not, in general, a polysymplectic manifold and the statement \ref{guth} is not
 true, in general (note that $T_{\pi_\mu(x_1,x_2)}M_\mu\cong\ds\frac{T_{(x_1,x_2)}J^{-1}(\mu)}{T_{(x_1,x_2)}
(G_\mu\cdot (x_1,x_2))}$ for $(x_1,x_2)\in J^{-1}(\mu)$).

As a consequence, we see that the generalization of the Marsden-Weinstein reduction theorem
 to the polysymplectic setting is not straightforward,
and some additional technical conditions must be added to the usual hypothesis.

\begin{remark}\label{exampleintro}
{\rm
 Note that the quotient vector space
$\ds\frac{T_{(x_1,x_2)}J^{-1}(\mu)}{T_{(x_1, x_2)}J^{-1}(\mu)\cap T_{(x_1, x_2)}^{\bot, 2}J^{-1}(\mu)}$
 is polysymplectic. In addition, using (\ref{condfin}),  we have that
 \[
    T_{(\pi_{\mu_1}(x_1),\pi_{\mu_2}(x_2))}\Big(\tilde{J}^{-1}(\mu_1)/G_{\mu_1}
 \times \tilde{J}^{-1}(\mu_2)/G_{\mu_2}\Big) \cong \ds\frac{T_{(x_1,x_2)}J^{-1}(\mu)}
{T_{(x_1, x_2)}J^{-1}(\mu)\cap T_{(x_1, x_2)}^{\bot, 2}J^{-1}(\mu)}
 \]
 for $(x_1,x_2)\in J^{-1}(\mu)= \tilde{J}^{-1}(\mu_1)\times \tilde{J}^{-1}(\mu_2)$,
 where $\pi_{\mu_i}\colon \tilde{J}^{-1}(\mu_i)\to \tilde{J}^{-1}(\mu_i)/G_{\mu_i}$
 is the canonical projection, $i\in \{1,2\}$.

 Thus, $\tilde{J}^{-1}(\mu_1)/G_{\mu_1} \times \tilde{J}^{-1}(\mu_2)/G_{\mu_2}$
 is a polysymplectic manifold (in fact, it is the product of the two reduced symplectic manifolds
$\tilde{J}^{-1}(\mu_1)/G_{\mu_1}$ and $\tilde{J}^{-1}(\mu_2)/G_{\mu_2}$).

 In the following Section \ref{main_reduction}, we develop a Marsden-Weinstein reduction procedure for
polysymplectic manifolds in such a way that when we apply this procedure to the polysymplectic manifold
$M=N\times N$ the resultant reduced polysymplectic manifold is just
$\tilde{J}^{-1}(\mu_1)/G_{\mu_1} \times \tilde{J}^{-1}(\mu_2)/G_{\mu_2}$ (see Theorem \ref{MW poly th} and Example \ref{ex0}).
}
\end{remark}

\section{Reduction of polysymplectic manifolds}
\label{main_reduction}

   The general setting of symplectic reduction (going back to E. Cartan) is the following (see \cite{AM-1978}, pag 298):
        \begin{quote}
            \textit{``Suppose that $M$ is a manifold and $\omega$ is a closed $2$-form on $M$; let
                \[
                    \ker\,\omega=\{v\in TM\  | \ \imath_v\omega=0\}
                \]
            the \textit{characteristic distribution} of $\omega$ and call $\omega$ regular if
            $\ker\,\omega$ is a subbundle of $TM$. In the regular case,
we note that $\ker\,\omega$ is an involutive distribution.
            By Frobenius's theorem $\ker\,\omega$ is integrable and hence it defines a foliation
$\mathcal{F}$ on $M$.
            Form the quotient space $M/\mathcal{F}$ by identification of all points on a leaf.
            Assume now that $M/\mathcal{F}$ is a manifold, the canonical projection
$M\to M/\mathcal{F}$ being a submersion.
            Then, the tangent space at a point ${\pi_\mu(x)}$ is isomorphic to $T_xM/\ker\,\omega(x)$ and hence
            $\omega$ projects on a well-defined closed, nondegenerate $2$-form on $M/\mathcal{F}$;
            that is, $M/\mathcal{F}$ is a symplectic manifold."}
        \end{quote}

    Marsden and Weinstein \cite{MW-1974} apply this general result to the case of submanifolds defined by
 the level sets of a
    Coad-equivariant momentum mapping of a given symplectic action.

    The aim of this section it to  extend these results to  polysymplectic manifolds, that is,
    we want define quotients of polysymplectic manifolds which inherit the respective structure in a
way analogous to the Marden-Weinstein reduction for a symplectic manifold.

    \subsection{Polysymplectic reduction by a submanifold}

        Using Frobenius' theorem and the fact that the family of $2$-forms associated with a polysymplectic structure are closed, we obtain the following lemma:
        \begin{lemma}
            Let $(M,\omega^1,\ldots, \omega^k)$ be a polysymplectic manifold and
 $\mathcal{S}$ be a submanifold of $M$
            with injective immersion $i\colon \mathcal{S}\to M$. If the distribution on
$\mathcal{S}$ given by \(\displaystyle\bigcap_{A=1}^k \ker\,(i^*\omega^A) \)
has constant rank then it defines a foliation $\mathcal{F}_\mathcal{S}$ on $\mathcal{S}$.
        \end{lemma}

        \begin{remark}
        {\rm
            Note that for each $x\in \mathcal{S}$, the following relations holds (see \cite{LV-2012})
            \[
                \displaystyle\bigcap_{A=1}^k \ker\,(i^*\omega^A)(x) =
T_x\mathcal{S} \cap T_x^{\bot, k}\mathcal{S}\,.
            \]
        }
        \end{remark}
        \begin{theorem}\label{reduc_submanifold}
            Let $(M,\omega^1,\ldots, \omega^k)$ be a polysymplectic manifold and let $\mathcal{S}$ be a
             submanifold of $M$ with injective immersion $i\colon \mathcal{S}\to M$. Assume that
                \begin{itemize}
                    \item The distribution $\displaystyle\bigcap_{A=1}^k \ker\,(i^*\omega^A)$ has constant rank,
                    \item The quotient space $\mathcal{S}/\mathcal{F}_\mathcal{S}$
is a manifold and the canonical projection
                    $\pi\colon \mathcal{S}\to \mathcal{S}/\mathcal{F}_\mathcal{S}$ is a submersion.
                \end{itemize}
            Then, there exists a unique polysymplectic structure
     $(\omega^1_{\mathcal{S}},\ldots, \omega^k_{\mathcal{S}})$ on $\mathcal{S}/\mathcal{F}_\mathcal{S}$
            such that, for every $A=1,\ldots, k$ the following relation holds:
                \[
                    \pi^*\omega^A_{\mathcal{S}}=i^*\omega^A\,.
                \]
        \end{theorem}
        \begin{proof}
            If $x$ is a point of $\mathcal{S}$, then the tangent space
$T_{\pi(x)}\left(\mathcal{S}/\mathcal{F}_\mathcal{S}\right)$ to
            $\mathcal{S}/\mathcal{F}_\mathcal{S}$ at the point $\pi(x)$ is isomorphic to the quotient space
            $T_x\mathcal{S}/\mathcal{F}_{\mathcal{S}}(x)$.

            Using that the $2$-form $\omega^A$ is closed, for $A\in \{1,\ldots, k\}$, we deduce that $i^*\omega^A$ is basic with respect to the foliation $\mathcal{F}_S$. This implies that every $i^*\omega^A$ will project on a well-defined $2$-form $\widetilde{\omega}^A_s$ on $S/\mathcal{F}_S$ such that $\pi^*\widetilde{\omega}^A_s=i^*\omega^A$.

%
%

            Finally, we will prove that $\ds\bigcap_{A=1}^k\ker\,\widetilde{\omega}^A_{\mathcal{S}}=0$. Let $[v_x]=T_x\pi(v_x)\in
            T_{\pi(x)}\left(\mathcal{S}/\mathcal{F_S}\right)$ be such
            that
                \[
                    \imath_{[v_x]}\widetilde{\omega}^A_\mathcal{S}(\pi(x))=0 \,.
                \]
            Furthermore, if $w_x\in T_x\mathcal{S}$ we obtain that
                \[
                    \begin{array}{lcl}
                        \left(\imath_{v_x}(i^*\omega^A)(x)\right)(w_x) &=&(i^*\omega^A)(x)(v_x,w_x)=
                            \left(\pi^*\widetilde{\omega}^A_\mathcal{S}\right)(x)(v_x,w_x) \\\noalign{\medskip}
                        &=& \widetilde{\omega}^A_\mathcal{S}(\pi(x))(T_x\pi(v_x),T_x\pi(w_x)) =
                         \left(\imath_{[v_x]}\widetilde{\omega}^A_\mathcal{S}(\pi(x))\right)([w_x])=0\,.
                    \end{array}
                \]
   Thus,
                \[
                    v_x\in \ds\bigcap_{A=1}^k \ker\,(i^*\omega^A)(x)\,,
                \]
            that is, $v_x$ is tangent to $\mathcal{F}_\mathcal{S}$ and
            then $[v_x]=T_x\pi(v_x)=0$.
        \end{proof}

\subsection{Marsden-Weinstein reduction for polysymplectic manifolds}\label{MWPR}

    In this section we apply the above general result to the
    case of  submanifolds defined  as the level sets of a $Coad^k$-equivariant
    momentum mapping of a given polysymplectic action. Our
    formulation follows the scheme of Marsden and Weinstein  \cite{MW-1974}.

Throughout this section we consider a polysymplectic Hamiltonian $G$-space
        $(M,\omega^1,\ldots, \omega^k; \Phi,J)$.

        The aim of this section is to impose conditions that
        guarantee that $J^{-1}(\mu)/G_\mu$ is a quotient manifold with a polysymplectic
        structure $(\omega_\mu^1,\ldots,\omega_\mu^k)$.

        As a consequence of a  well-known result, one obtains:
        \begin{lemma}
            Let $(M,\omega^1,\ldots, \omega^k; \Phi,J)$ be a polysymplectic Hamiltonian $G$-space.
            If $\mu=(\mu_1,\ldots, \mu_k)\in \mathfrak{g}^*\times \stackrel{k}{\ldots}\times\mathfrak{g}^*$
            is a regular value of the momentum map $J\equiv (J^1,\ldots, J^k)$
            (by Sard's theorem, it takes place for ``almost all'' $\mu$), then
                \[
                    \mathcal{S}=J^{-1}(\mu)=J^{-1}(\mu_1,\ldots,\,u_k)
                \]
            is a regular submanifold of $M$.
        \end{lemma}

        Therefore, we can apply the general theorem of polysymplectic reduction
        (see Theorem \ref{reduc_submanifold}) by a submanifold with $\mathcal{S}=J^{-1}(\mu)$,
 and we obtain the following

        \begin{theorem}
             Let $(M,\omega^1,\ldots, \omega^k; \Phi,J)$ be a polysymplectic Hamiltonian $G$-space and
             $\mu=(\mu_1,\ldots, \mu_k)\in \mathfrak{g}^*\times \stackrel{k}{\ldots}\times\mathfrak{g}^*$
             be a regular value of the momentum map $J\equiv (J^1,\ldots, J^k)$.
We denote by $i\colon \mathcal{S}=J^{-1}(\mu)\to M$
             the canonical inclusion.
             Let us assume that:
                \begin{itemize}
     \item The distribution $\displaystyle\bigcap_{A=1}^k\ker\, (i^*\omega^A)$ has constant rank (we denote by
                    $\mathcal{F}_{J^{-1}(\mu)}$ the induced foliation),
                    \item $J^{-1}(\mu)/\mathcal{F}_{J^{-1}(\mu)}$ is a manifold and the canonical projection
                    $\pi_\mu\colon J^{-1}(\mu)\to  J^{-1}(\mu)/\mathcal{F}_{J^{-1}(\mu)}$ is a
                    submersion.
                \end{itemize}
             Then there exists an unique polysymplectic structure $(\omega_\mu^1,\ldots, \omega_\mu^k)$
 on $J^{-1}(\mu)/\mathcal{F}_{J^{-1}(\mu)}$
             such that the following relationship holds for every $A=1,\ldots, k$
                \[
                    \pi_\mu^*\omega^A_\mu =i^*\omega^A\,.
                \]
        \end{theorem}

        Now we seek conditions, expressed in terms of the elements
        of the polysymplectic Hamiltonian $G$-space $(M,\omega^1,\ldots, \omega^k;
        \Phi,J)$, such that the two assumptions made in the previous
        theorem are satisfied. The  first point is to study the following question:
        \begin{center}
            \textit{Under what conditions does the distribution
$\displaystyle\bigcap_{A=1}^k\ker\, (i^*\omega^A)$ have constant rank?}.
        \end{center}
        Now we study this question, giving
        conditions that guarantee $\displaystyle\bigcap_{A=1}^k\ker\, (i^*\omega^A)=T_x(G_\mu \cdot x)$,
for every  $x\in J^{-1}(\mu)$, and assuming that the action of $G_\mu$ on $J^{-1}(\mu)$ is free.
In such a case, the leaves of the induced foliation $\mathcal{F}_{J^{-1}(\mu)}$
 are the orbits of the action of $G_\mu$ on $J^{-1}(\mu)$.

        \begin{lemma}\label{isotropy} Let $\mu=(\mu_1,\ldots, \mu_k)$ be a regular
        value of $J$.
            \begin{enumerate}
                \item If $G_{\mu_A}$ denotes the isotropy subgroup of $G$ under the coadjoint action
                $Coad$ at $\mu_A\in\mathfrak{g}^*$ and
                $\mathfrak{g}_{\mu_A}$ its Lie algebra, then
                    \[
                        G_\mu=G_{(\mu_1,\ldots,\mu_k)}=\displaystyle\bigcap_{A=1}^kG_{\mu_A}\quad
                        \makebox{and} \quad \mathfrak{g}_\mu=\mathfrak{g}_{(\mu_1,\ldots,\mu_k)}=
                        \displaystyle\bigcap_{A=1}^k\mathfrak{g}_{\mu_A}\,.
                    \]
                \item $G_\mu$ acts on $J^{-1}(\mu)$ and the orbit space $J^{-1}(\mu)/G_\mu$ is well-defined.
                \item For every $x\in J^{-1}(\mu)$,
                    \[
                        T_x(G_\mu\cdot x)\subseteq \displaystyle\bigcap_{A=1}^k\ker\,(i^*\omega^A)(x) \ .
                    \]
            \end{enumerate}
        \end{lemma}
        \begin{proof}
            \begin{enumerate}
                \item Using (\ref{coad^k}), one obtains:
                    \[
                    \begin{array}{lcl}
                        G_\mu&=&\{g\in G \ \mid \ Coad^k_g(\mu)=\mu\}= \{g\in G \ \mid \ Coad_g(\mu_A)=
                        \mu_A\,,{\rm for}\  A=1,\ldots, k\}\\\noalign{\medskip}
                         &=& \displaystyle\bigcap_{A=1}^k\{g\in G \ \mid \ Coad_g(\mu_A)=\mu_A\} =
\displaystyle\bigcap_{A=1}^kG_{\mu_A} \ .
                    \end{array}
                    \]
                As a consequence of this identity,
 it is immediate to prove the analogous relationship among the  Lie  algebras.

                \item From the polysymplectic action $\Phi\colon
                G\times M\to M$, we define the action
                    \[
                        \begin{array}{rccl}
                            \Phi_\mu\colon & G_\mu\times J^{-1}(\mu) & \to & J^{-1}(\mu)\\\noalign{\medskip}
                             & (g,x) & \mapsto & \Phi_{\mu}(g,x)\colon = \Phi(g,x)
                        \end{array}
                    \]
                This is a well-defined map. Indeed, let
                $(g,x)\in G_\mu\times J^{-1}(\mu)\subset G\times
                M$, then as $J$ is $Coad^k$-equivariant, we have:
                    \[
                        J(\Phi_\mu(g,x))=J(\Phi(g,x))= Coad^k_g(J(x)) = Coad^k_g(\mu)=\mu\,.
                    \]
                Therefore, if $(g,x)\in G_\mu\times J^{-1}(\mu)$
                then $\Phi_\mu(g,x)\in J^{-1}(\mu)$.

                \item We consider the action $\Phi_\mu\colon G_\mu\times J^{-1}(\mu)\to J^{-1}(\mu)$.
                If $\mathfrak{g}_\mu$ is the Lie algebra of
                $G_\mu$ we have
                    \[
                        T_x(G_\mu\cdot x) =\{\xi_{J^{-1}(\mu)}(x) \ \mid \  \xi\in\mathfrak{g}_\mu\}\,.
                    \]

                If $\xi_{J^{-1}(\mu)}(x)\in T_x(G_\mu \cdot x)$, then
                $\xi_{J^{-1}(\mu)}(x)\in\displaystyle\bigcap_{A=1}^k\ker \,(i^*\omega^A)(x)$ if, and only if,
                 \[
                 (i^*\omega^A)(x)\left(\xi_{J^{-1}(\mu)}(x),X_x\right)=0 \] for every $X_x\in T_x(J^{-1}(\mu))$.
	Now, we have
                    \[
                        \begin{array}{lcl}
                            &&(i^*\omega^A)(x)\left(\xi_{J^{-1}(\mu)}(x),X_x\right) = \omega^A(x)(\xi_M(x),X_x)=
                            (\imath_{\xi_M}\omega^A)(x)(X_x)
                            = (d\hat{J}^A_\xi)(x)(X_x)=X_x(\hat{J}^A_\xi)\,.
                        \end{array}
                    \]
                But as $X_x\in T_x(J^{-1}(\mu))$, we obtain that
                    \[
                        0=T_xJ(X_x)=(T_xJ^1(X_x),\ldots, T_xJ^k(X_x))\,,
                    \]
                and thus, $T_xJ^A (X_x)=0$. Therefore, for
                $\xi\in\mathfrak{g}$ we have
                    \[
                        \left( T_xJ^A (X_x)\right)(\xi)=0\, ;
                    \]
                that is, $X_x(\hat{J}^A_\xi)=0$.
            \end{enumerate}
        \end{proof}

        From this lemma we obtain that,
            \[
                T_x(G_\mu\cdot x)\subseteq \displaystyle\bigcap_{A=1}^k\ker\,(i^*\omega^A)(x)=
T_x(J^{-1}(\mu))\cap T_x^{\bot, k}(J^{-1}(\mu))\,\quad \makebox{for every $x\in J^{-1}(\mu)$,}
            \]
        but, in general, the condition
            \begin{equation}\label{converse}
                \displaystyle\bigcap_{A=1}^k\ker\,(i^*\omega^A)(x)\subseteq T_x(G_\mu\cdot x)
            \end{equation}
        does not hold. Note that if (\ref{converse}) holds and the action of $G_\mu$ on $J^{-1}(\mu)$
is free then the distribution
        $\displaystyle\bigcap_{A=1}^k\ker\,(i^*\omega^A)$ has constant rank.
        In addition, if the action of $G_\mu$ on $J^{-1}(\mu)$ is proper, then $J^{-1}(\mu)/G_\mu$
        is a quotient manifold which admits a polysymplectic structure. In fact,
            \[
                J^{-1}(\mu)/G_\mu = J^{-1}(\mu)/\mathcal{F}_{J^{-1}(\mu)}\,.
            \]

        So, a new natural question arises:
            \begin{center}
            {\textit{Under what conditions can it be assured that
$T_x(G_\mu\cdot x)= \displaystyle\bigcap_{A=1}^k\ker\,(i^*\omega^A)(x)$,
            for every $x\in J^{-1}(\mu)$? }}
            \end{center}

        Now we give conditions that guarantee that
            \[
                T_x(G_\mu\cdot x)= \displaystyle\bigcap_{A=1}^k\ker\,(i^*\omega^A)(x) ,
 \makebox{ for every} \; x\in J^{-1}(\mu) \ ,
            \]
        which implies that  $T_x(J^{-1}(\mu))/T_x(G_\mu\cdot x)$ is a polysymplectic vector space.

        First, we recall  the following immediate result, which is fundamental in our description.
            \begin{lemma}\label{algebra}
                Let $\Pi_A\colon V\to V_A$ be $k$
                epimorphisms of real vector spaces of finite
                dimension. Assume that there exists a symplectic
                structure $\omega^A$ on $V_A$ for each index $A$ and
                $\bigcap_{A=1}^k\ker\, \Pi_A=\{0\}$, then
                $(V,\Omega_1,\ldots, \Omega_ k)$, with
                $\Omega_A=\Pi_A^*\omega^A$ is a
                polysymplectic vector space.
             \end{lemma}

        We consider again the example described in Section \ref{Gu-error}
(see Remark \ref{exampleintro}). In this example, the reduced polysymplectic manifold is the product of two
 reduced symplectic manifolds: $\tilde{J}^{-1}(\mu_1)/G_{\mu_1}$ and $\tilde{J}^{-1}(\mu_2)/G_{\mu_2}$.
Using this fact for each $(x_1,x_2)\in J^{-1}(\mu)$ we can obtain the reduced polysymplectic structure
by applying Lemma \ref{algebra} as follows
        \[
            V=T_{(\pi_{\mu_1}(x_1),\pi_{\mu_2}(x_2))}\left(\tilde{J}^{-1}(\mu_1)/G_{\mu_1}
\times \tilde{J}^{-1}(\mu_2)/G_{\mu_2}
             \right)
        \]
        and
        \[
            V_A= T_{\pi_{\mu_A}(x_A)}\left(\tilde{J}^{-1}(\mu_A)/G_{\mu_A}\right) =
T_{x_A}(\tilde{J}^{-1}(\mu_A)) / T_{x_A}(G_{\mu_A}\cdot x_A)\,.
        \]

        Observe that the vector spaces $V_A$  can be described as the quotients
        \[
            V_A=\displaystyle\frac{\left(\displaystyle
            \frac{\ker\, T_{(x_1,x_2)} J^A}{\ker\,\omega^A(x_1,x_2)}\right)}
            {\{[\xi_M(x_1,x_2)]\,\mid\,\xi\in\mathfrak{g}_{\mu_A}\}}
        \]
        where $J^A=\widetilde{J}$, for $A\in\{1,2\}$,  $\ker\,\omega^1(x_1,x_2)=\{0\}\times T_{x_2}N$
and $\ker\,\omega^2(x_1,x_2)=T_{x_1}N\times \{0\}$.

        We  now return to the general case of a polysymplectic Hamiltonian $G$-space
$(M,\omega^1,\ldots, \omega^k;\Phi,J)$ and assume that $\mu=(\mu_1,\ldots, \mu_k)$
is a regular value of the momentum map
$J\colon M\to \mathfrak{g}^*\times\stackrel{k}{\ldots}\times \mathfrak{g}^*$.
Then, using that $J$ is a  momentum map, we deduce that $\ker\,\omega^A(x)$
 is a subspace of $\ker\, T_xJ^A$. In fact, if $X\in \ker \omega^A(x)$ and $\xi\in \mathfrak{g}$,
 we have that $\{ (T_xJ^A)(X)\}(\xi)= d\widetilde{J}^A_\xi(X) =
 (\imath_{\xi_M}\omega^A)(x)(X)= -(\imath_X\omega^A)(x)(\xi_M(x))=0$.
On the other hand, since $G_{\mu_A}$ acts on $(J^A)^{-1}(\mu_A)$,
if $x\in (J^A)^{-1}(\mu_A)$, it follows that
$\{\xi_M(x)\,\vert\, \xi\in\mathfrak{g}_{\mu_A}\}$
is also a subspace of
$\ker\, T_xJ^A$. Thus, if $pr^M_A\colon T_xM\to \ds\frac{T_xM}{\ker\, \omega^A(x)}$
is the canonical projection, we have that
$pr^M_A(\{\xi_M(x)\,\vert\, \xi\in\mathfrak{g}_{\mu_A}\} =  \{[\xi_M(x)] \,\vert\, \xi\in\mathfrak{g}_{\mu_A}\}$
is a subspace of $\ds\frac{\ker T_xJ^A}{\ker\,\omega^A(x)}$.
Therefore, as in the previous example, we can consider the quotient space
        \[
                    V_A=\displaystyle\frac{\left(\displaystyle
                    \frac{\ker\, T_xJ^A }{\ker\,\omega^A(x)}\right)}
                    {\{[\xi_M(x)]\,\mid\,\xi\in\mathfrak{g}_{\mu_A}\}}
                \]

        Thus, the problem of finding conditions that
        guarantee that $T_x(G_\mu\cdot x)= $ $ \displaystyle\bigcap_{A=1}^k{\rm
        ker}\,(i^*\omega^A)(x)$ can be decomposed in two steps:
        \begin{enumerate}
             \item
             To prove that, for every $x\in J^{-1}(\mu)$, the vector space
                \[
                    V_A=\displaystyle\frac{\left(\displaystyle
                    \frac{\ker\, T_xJ^A }{\ker\,\omega^A(x)}\right)}
                    {\{[\xi_M(x)]\,\mid\,\xi\in\mathfrak{g}_{\mu_A}\}}
                \]
                is a symplectic vector space, where $[\xi_M(x)]=pr_A^M(\xi_M(x))$ and
                $pr_A^M\colon T_xM\to \displaystyle\frac{T_xM}{\ker\,\omega^A(x)}$
                is the canonical projection.
              \item
              To find conditions guaranteeing that we can define $k$ linear epimorphisms
                \[
                    \widetilde{\pi}^A_x: T_{{\pi_\mu(x)}}(J^{-1}(\mu)/G_\mu)\longrightarrow
                    \ds\frac{\left(\ds\frac{\ker\,  T_xJ^A  }{\ker\, \omega^A(x)}\right) }
{\{ [\xi_M(x)]\,\mid\,\xi\in \mathfrak{g}_{\mu_A} \} }
                \]
              such that $\displaystyle \bigcap_{A=1}^k\ker\,\widetilde{\pi}^A_x=\{ 0\}$.
\end{enumerate}
              We see that these conditions also imply that
$T_x(G_\mu\cdot x)=  \displaystyle\bigcap_{A=1}^k\ker \,(i^*\omega^A)(x)$.

\paragraph{$\bullet\; $ {\sc Step 1}}\

            As mentioned above, our aim is to prove the following proposition
            \begin{prop}\label{step1}
                Let $(M,\omega^1,\ldots,\omega^k;\Phi,J)$ be a polysymplectic Hamiltonian $G$-space and
                $\mu=(\mu_1,\ldots, \mu_k)\in \mathfrak{g}^*\times\stackrel{k}{\ldots}\times\mathfrak{g}^*$
 be a regular value of $J$,
                then for $A=1,\ldots,k$ and $x\in J^{-1}(\mu)$ we have that
                    \[
                        \displaystyle\frac{\left(\displaystyle
                        \frac{\ker\, T_xJ^A }{\ker\,\omega^A(x)}\right)}
                        {\{[\xi_M(x)]\,\mid\,\xi\in\mathfrak{g}_{\mu_A}\}}
                    \]
                is a symplectic vector space.
            \end{prop}

            The idea of the proof is to obtain
            a family of closed 2-forms in the different quotient spaces of the following diagram
            (based on   Marsden-Weinstein's reduction procedure):
            \[
 \xymatrix@R=1.3cm{   \ker \,  T_xJ^A
 {\;,\; {\omega_{J^A}(x)}}   \ar[r]^-{i^A_x}\ar[d]^-{pr^{J^A}}& T_xM\;,\; {\omega^A(x)}\ar[d]^-{pr^M_A}\\
            \ds\frac{\ker \,  T_xJ^A  }{\ker\, \omega^A(x)}
{\;,\; {\widetilde{\omega_{J^A}(x)}}} \ar[r]^-{\widetilde{i^A_x}}\ar[d]^-{\widetilde{pr_A}} &
             \ds\frac{T_xM}{\ker\, \omega^A(x)} {\;,\; {\widetilde{\omega^{A}(x)}}}\\
   \ds\frac{\left(\ds\frac{\ker \,  T_xJ^A  }{\ker\, \omega^A(x)}\right)}{ \{ [ \xi_M(x) ]  \,\mid\,
    \xi\in \mathfrak{g}_{\mu_A} \} }  {\;,\; { \omega_{\mu_A}(x) } }&}
            \]

Before proving this proposition, we first need some lemmas in which we assume the same hypothesis
as in  Proposition \ref{step1}. The first is a straightforward consequence of the definition
of a symplectic form on a vector space and the definition of $\ker\,\omega^A(x)$.

            \begin{lemma}\label{tildeomega}
                For every $A=1,\ldots, k$, there exists a unique symplectic form
                $\widetilde{\omega^A(x)}$ on $\displaystyle\frac{T_xM}{\ker\, \omega^A(x)} $ such that
                    \[
                        [pr^M_A]^*[\widetilde{\omega^A(x)}]= \omega^A(x)\,.
                    \]
            \end{lemma}

    Now we consider the quotient space $\displaystyle\frac{\ker \,  T_xJ^A  }{\ker\, \omega^A(x)}$, and the
           vectorial subspaces of $\displaystyle\frac{T_xM}{\ker\,\omega^A(x)}$ defined by
             $\{[\xi_M(x)]\,\mid\,\xi\in\mathfrak{g}\}$ and $\{[\xi_M(x)]\,\mid\,\xi\in\mathfrak{g}_{\mu_A}\}$
            which satisfy the following properties:

            \begin{lemma}\label{properties}\

                \begin{enumerate}
                  \item $ \{[\xi_M(x)]\,\mid\,\xi\in\mathfrak{g}_{\mu_A}\} =
             \{[\xi_M(x)] \,\mid\,\xi\in\mathfrak{g}\} \cap \displaystyle\frac{\ker \,  T_xJ^A  }{\ker\, \omega^A(x)}$.
 \item $\displaystyle\frac{\ker \,  T_xJ^A  }{\ker\, \omega^A(x)}= \{[\xi_M(x)] \,\mid\,\xi\in\mathfrak{g}\}^\perp$,
where the symbol  $^\perp$ denotes the symplectic orthogonal in
$\displaystyle\frac{T_xM}{\ker\,\omega^A(x)}$ with
                     respect to $\widetilde{\omega^A(x)}$.
                    \item $\left[\displaystyle\frac{\ker \, T_xJ^A  }{\ker\, \omega^A(x)}\right]^\perp=
\{[\xi_M(x)] \,\mid\,\xi\in\mathfrak{g}\}$.
                    \item $ \{[\xi_M(x)]\,\mid\,\xi\in\mathfrak{g}_{\mu_A}\} =
  \displaystyle\frac{\ker \,  T_xJ^A  }{\ker\, \omega^A(x)} \cap \left[\displaystyle\frac{\ker \,  T_xJ^A  }
{\ker\, \omega^A(x)}\right]^\perp$.
                \end{enumerate}
            \end{lemma}
            \begin{proof}
                \begin{enumerate}
                                      \item
                    The proof of this item is similar to the proof of item $(i)$ of Lemma 4.3.2
                    in \cite{AM-1978}.
                    \item
 Taking into account that $\ker\,\omega^A(x)\subseteq\ker\,  T_xJ^A $, the proof of this item is similar to the
                    proof of item $(ii)$ of Lemma 4.3.2 in \cite{AM-1978}.                    \item
                    It is a consequence of $(2)$, since $\widetilde{\omega^A(x)}$ is symplectic.
                    \item
                    It is a consequence of items $(1)$ and $(3)$ of this lemma.
                \end{enumerate}
            \end{proof}

            \begin{lemma}
                Let $(M,\omega^1,\ldots,\omega^k;\Phi,J)$ be a
                polysymplectic Hamiltonian $G$-space,
                $\mu=(\mu_1,\ldots, \mu_k)\in \mathfrak{g}^*\times \stackrel{k}{\ldots}\times\mathfrak{g}^*$
           be  a regular value of $J$ and $\widetilde{\omega^A(x)}$
  the symplectic structure  on $\ds\frac{T_xM}{\ker\,\omega^A(x)}$ defined in Lemma \ref{tildeomega}.
                Then there exists a skew-symmetric bilinear form
    $\widetilde{\omega_{J^A}(x)}$ on $\displaystyle\frac{\ker \,  T_xJ^A  }{\ker\, \omega^A(x)}$ such that
                    \[
                        [pr^{J^A}]^*\widetilde{\omega_{J^A}(x)}=
                        \omega_{J^A}(x)\,,
                    \]
                where $pr^{J^A}\colon \ker\,  T_xJ^A  \to\displaystyle\frac{\ker\,  T_xJ^A } {\ker\,\omega^A(x)}$
                is the canonical projection,
 $i^A_x\colon \ker\,  T_xJ^A  \to T_xM$ is the canonical inclusion and
$\omega_{J^A}(x)\colon=(i^A_x)^*[\omega^A(x)]$.
Moreover,  taking the inclusion
                $\widetilde{i^A_x}\colon \displaystyle\frac{\ker\, T_xJ^A }{\ker\,\omega^A(x)}\longrightarrow
                 \displaystyle\frac{T_xM}{\ker\,\omega^A(x)}$,
                the following  relation holds:
                    \[
                        \widetilde{\omega_{J^A}(x)}=[\widetilde{i^A_x}]^*[\widetilde{\omega^A(x)}]
                    \]
            \end{lemma}
            \begin{proof}
                Consider the $2$-form on $\ker\,  T_xJ^A $ defined by
                    \[
                        \omega_{J^A}(x) = (i^A_x)^* [\omega^A(x)] \ ;
                    \]
                that is, if $v_x,w_x\in \ker\,  T_xJ^A $ then
                    \[
                        \omega_{J^A}(x)(v_x,w_x)=\omega^A(x)(i^A_x(v_x),i^A_x(w_x))=\omega^A(x)(v_x,w_x)\,.
                    \]
                Taking into account that $ \ker\,\omega^A(x)\subseteq \ker\, T_xJ^A$, it is easy to prove that
                    \begin{equation}\label{ker}
                        \ker\,\omega^A(x)\subseteq \ker\,
                        \omega_{J^A}(x)\,.
                    \end{equation}
                As (\ref{ker}) holds, $\omega_{J^A}(x)$ induces a
                well-defined $2$-form $\widetilde{\omega_{J^A}(x)}$
                on the vector space $\displaystyle\frac{\ker\, T_xJ^A }
                {\ker\,\omega^A(x)}$.  Furthermore, it is clear that
$[pr^{J^A}]^*\widetilde{\omega_{J^A}(x)}=\omega_{J^A}(x)$.
                Moreover, $\widetilde{\omega_{J^A}(x)}$ is the
                restriction of $\widetilde{\omega^A(x)}$ to the
                subspace $\displaystyle\frac{\ker\,  T_xJ^A }
                {\ker\,\omega^A(x)}$. Indeed, by definition,
                $\widetilde{\omega_{J^A}(x)}$ is characterized by $[pr^{J^A}]^*\widetilde{\omega_{J^A}(x)}=
                        \omega_{J^A}(x)$. In addition, we have that
                    \[\begin{array}{lcl}
                        \omega_{J^A}(x)=(i^A_x)^*(\omega^A(x))&=&
                        (i^A_x)^*\left((pr^M_A)^*\widetilde{\omega^A(x)}
                        \right) = (pr^M_A\circ
                        i^A_x)^*(\widetilde{\omega^A(x)})\\\noalign{\medskip}
                        &=& (\widetilde{i^A_x}\circ\, pr^{J^A})^*
                        (\widetilde{\omega^A(x)}) =
                        (pr^{J^A})^*\left((\widetilde{i^A_x})^*(\widetilde{\omega^A(x)})\right)\, ,
                    \end{array}\]
                and then $\widetilde{\omega_{J^A}(x)}=(\widetilde{i^A_x})^*(\widetilde{\omega^A(x)})$.
            \end{proof}

            Now, as a consequence of the above lemmas, we are able to prove Proposition  \ref{step1}.
            \begin{proof}[\textit{Proof of Proposition \ref{step1}}]
                We have that
                $\{[\xi_M(x)]\,\mid\,\xi\in\mathfrak{g}_{\mu_A}\}$
                is a subspace of
                $\displaystyle\frac{\ker \,  T_xJ^A  }{\ker\,
                \omega^A(x)}$. Then, we can consider the quotient
                vector space $\displaystyle
                        \displaystyle\frac{\left(\displaystyle
                        \frac{\ker\, T_xJ^A }{\ker\,\omega^A(x)}\right)}
                        {\{[\xi_M(x)]\,\mid\,\xi\in\mathfrak{g}_{\mu_A}\}}$,
                with canonical projection
                $$
                \widetilde{pr_A}\colon \displaystyle\frac{\ker \, T_xJ^A  }{\ker\,
                \omega^A(x)}\longrightarrow \displaystyle\frac{\left(\displaystyle
                        \frac{\ker\, T_xJ^A }{\ker\,\omega^A(x)}\right)}
                        {\{[\xi_M(x)]\,\mid\,\xi\in\mathfrak{g}_{\mu_A}\}} \ .
	$$
                Now, using  item $(4)$ in Lemma \ref{properties},
                it is easy to prove that
                $\widetilde{\omega_{J^A}(x)}$ induces a well-defined
                non-degenerate  $2$-form $\omega_{\mu_A}(x)$ on $\displaystyle\frac{\left(\displaystyle
                        \frac{\ker\, T_xJ^A }{\ker\,\omega^A(x)}\right)}
                        {\{[\xi_M(x)]\,\mid\,\xi\in\mathfrak{g}_{\mu_A}\}}$
                        given by
                    \[
                        \omega_{\mu_A}(x)(\big[[v_x]\big],\big[[w_x]\big])
                        \colon
                        =\widetilde{\omega_{J^A}(x)}
                        ([v_x],[w_x])\,,\quad \makebox{for\ } [v_x],[w_x]\in \displaystyle
                        \frac{\ker\,
                         T_xJ^A }{\ker\,\omega^A(x)}\,.
                    \]
            \end{proof}

\paragraph{$\bullet\; $ {\sc Step 2}}\

        In this step  we assume that the action of $G_\mu$ on $J^{-1}(\mu)$ is free and proper, and
        thus $J^{-1}(\mu)/G_\mu$ is a quotient manifold. Then we can define $k$ linear morphisms
                \[
                    \widetilde{\pi}^A_x: T_{{\pi_\mu(x)}}(J^{-1}(\mu)/G_\mu) \longrightarrow
                    \ds\frac{\left(\ds\frac{\ker\,  T_xJ^A  }{\ker\, \omega^A(x)}\right) }{\{ [\xi_M(x)]
\,\mid\, \xi\in \mathfrak{g}_{\mu_A} \} }
                \]
       In fact:

        \begin{prop}\label{tildepi}
            Let $(M,\omega^1,\ldots,\omega^k;\Phi,J)$ be a polysymplectic Hamiltonian $G$-space and let
            $\mu=(\mu_1,\ldots, \mu_k)\in \mathfrak{g}^*\times\stackrel{k}{\ldots}\times\mathfrak{g}^*$
be a regular value of $J$.
            Suppose that $G_\mu$ acts freely and properly on $J^{-1}(\mu)$, then:
            \begin{enumerate}
                \item
                For every $x\in J^{-1}(\mu)$,
  $T_{\pi_{\mu}(x)}(J^{-1}(\mu)/G_\mu)\equiv \displaystyle\frac{T_x(J^{-1}(\mu))}{T_x(G_\mu\cdot x)}\equiv
   \displaystyle\frac{\bigcap_{A=1}^k\ker\, T_xJ^A }{\{\xi_{J^{-1}(\mu)}(x)\,\mid\,\xi\in \mathfrak{g}_\mu\}}$\,.
                \item
             There exists a linear map $\widetilde{\pi}^A_x$
                 between the quotient vector spaces
                 $T_{{\pi_\mu(x)}}(J^{-1}(\mu)/G_\mu)\equiv\ds\frac{T_x(J^{-1}(\mu))} {T_x(G_\mu\cdot x)}$ and
                 $\ds\frac{\left(\ds\frac{\ker\,  T_xJ^A  }{\ker\, \omega^A(x)}\right) }{\{ [\xi_M(x)]
\,\mid\, \xi\in \mathfrak{g}_{\mu_A} \} }$, for every $A=1,\ldots, k$.
                \end{enumerate}
        \end{prop}
        \begin{proof}
            \begin{enumerate}
                \item Using the fact that  $T_x(J^{-1}(\mu))=\bigcap_{A=1}^k\ker\, T_xJ^A $, we deduce the result.

                \item
                As $T_x(J^{-1}(\mu))=\bigcap_{B=1}^k\ker\, T_xJ^B$, then,
                for every $A=1,\ldots, k$ we have that $T_x(J^{-1}(\mu))\subseteq \ker\,  T_xJ^A $
                and therefore we can consider the composition
                    \[
   \xymatrix{ T_x(J^{-1}(\mu))\ar[r]^-{j_A} \ar@/^{9mm}/[rrr]^{\pi^A_x} & \ker\, T_xJ^A  \ar[rr]^-{pr^{J^A}}& &
                        \displaystyle\frac{\ker\,  T_xJ^A }{\ker\, \omega^A(x)} } \ .
                    \]

                    Moreover, as
                    $\mathfrak{g}_\mu=\bigcap_{A=1}^k\mathfrak{g}_{\mu_A}$
                    (see item $(1)$ in Lemma \ref{isotropy}), we have
                        \[
                            \pi^A_x\left(T_x(G_\mu\cdot x)\right) \subseteq
 \{[\xi_M(x)]\,\mid\,\xi\in\mathfrak{g}_{\mu_A}\}=\Big[\ds\frac{\ker\, T_xJ^A}{\ker\,\omega^A(x)}\Big]^\bot\,.
                        \]
                    Therefore (see Lemma \ref{properties}),
                        \[
                            \pi^A_x(\xi_{J^{-1}(\mu)}(x))\in \displaystyle\frac{\ker\, T_xJ^A }{\ker\, \omega^A(x)} \cap
 \left[\displaystyle\frac{\ker \, T_xJ^A  }{\ker\, \omega^A(x)}\right]^\perp=
 \{[\xi_M(x)]\,\mid\,\xi\in\mathfrak{g}_{\mu_A}\}\,.
                        \]
                     Hence, $\pi^A_x$ induces a well-defined linear map
                        \[
                        \widetilde{\pi}^A_x: \ds\frac{T_x(J^{-1}(\mu))} {T_x(G_\mu\cdot x)}\longrightarrow
\ds\frac{\left(\ds\frac{\ker\,  T_xJ^A  }{\ker\, \omega^A(x)}\right)}{\{ [\xi_M(x)]/\xi\in \mathfrak{g}_{\mu_A} \} }\ .
                        \]
            \end{enumerate}
        \end{proof}

	Further results require to prove the following lemma:
	
        \begin{lemma}\label{omegamua}
            Let $(M,\omega^1,\ldots, \omega^k; \Phi,J)$ be a polysymplectic Hamiltonian $G$-space
            and let $\mu=(\mu_1,\ldots, \mu_k)\in \mathfrak{g}^*\times \stackrel{k}{\ldots}\times\mathfrak{g}^*$
            be a regular value of  $J\equiv (J^1,\ldots, J^k)$.
            Let $i\colon J^{-1}(\mu)\to M$ be the canonical
            inclusion. Assume that $G_\mu$ acts freely and properly on
            $J^{-1}(\mu)$.

            For every $A=1,\ldots, k$, the $2$-form $i^*\omega^A$ on
            $J^{-1}(\mu)$ induces a closed  $2$-form $\omega^A_\mu$
             on $ J^{-1}(\mu)/G_\mu$ which satisfies the following properties:
                \begin{enumerate}
                    \item $\pi_\mu^*\omega^A_\mu=i^*\omega^A$,
where $\pi_\mu\colon J^{-1}(\mu) \to J^{-1}(\mu)/G_\mu$ is the canonical projection.
                    \item If $x\in J^{-1}(\mu)$ then $[\widetilde{\pi}^A_x]^*(\omega_{\mu_A}(x))=
\omega^A_\mu(\pi_\mu(x))$
                \end{enumerate}
        \end{lemma}
        \begin{proof}
        \begin{enumerate}
        \item If $x\in J^{-1}(\mu)$ we have that $T_x(G_\mu\cdot x)\subseteq \bigcap_{A=1}^k\ker\,
            (i^*\omega^A)(x)$, (see item $\mathbf{(3)}$ in lemma
            \ref{isotropy}) and, thus, $\imath_{\xi_{J^{-1}(\mu)}}(i^*\omega^A)=0$,
for $\xi\in \mathfrak{g}_\mu$. In addition, using that $i^*\omega^A$ is  a closed $2$-form,
we deduce that $i^*\omega^A$ is $\pi_\mu$-basic.
 Therefore, there exists a unique $2$-form $\omega^A_\mu$ on $J^{-1}(\mu)/G_\mu$
such that $\pi_\mu^*\omega^A_\mu=i^*\omega^A$. Moreover, since $\omega^A$ is a closed $2$-form, we have that $\omega^A_\mu$ is a closed $2$-form.

%
%

            \item If $[v_x]=T_x\pi_\mu(v_x)$  denotes
            the corresponding equivalence class in
            $\displaystyle\frac{T_x(J^{-1}(\mu))}{T_x(G_\mu\cdot x)}$ of
            $v_x\in T_x(J^{-1}(\mu))$, then
                \[\begin{array}{lcl}
                    \left([\widetilde{\pi}^A_x]^*(\omega_{\mu_A}(x))\right)([v_x],[w_x])
 &=& \left([\widetilde{\pi}^A_x]^*(\omega_{\mu_A}(x))\right)(T_x\pi_\mu(v_x),T_x\pi_\mu(w_x))\\
\noalign{\medskip}
                    &=& \omega_{\mu_A}(x)\left((\widetilde{\pi}^A_x\circ T_x\pi_\mu)(v_x),
                    (\widetilde{\pi}^A_x\circ T_x\pi_\mu)(w_x)\right)\\\noalign{\medskip}
                    &=& \omega_{\mu_A}(x)\left((\widetilde{pr_A}\circ \pi^A_x)(v_x),
(\widetilde{pr_A}\circ \pi^A_x)(w_x)\right)\\\noalign{\medskip}
                    &=& \widetilde{\omega_{J^A}(x)}(\pi^A_x(v_x),\pi^A_x(w_x))\\\noalign{\medskip}
                    &=& [\pi^A_x]^* \widetilde{\omega_{J^A}(x)}(v_x,w_x)=(i^*\omega^A)(x)(v_x,w_x)
                    \\\noalign{\medskip}
                    &=& \omega^A_\mu(\pi_\mu(x))([v_x],[w_x])
                    \,,
                \end{array}\]
            where we have used that
            $\widetilde{\pi}^A_x\circ T_x\pi_\mu=\widetilde{pr_A}\circ \pi^A_x$ and that
            $[\pi^A_x]^* \widetilde{\omega_{J^A}(x)}=(i^*\omega^A)(x)$. The following diagram illustrates the situation:
            \[
                \xymatrix{
                (T_x(J^{-1}(\mu)), i^*\omega^A(x))
                \ar[r]^-{j_A}
                \ar[d]_{T_x\pi_\mu}
                \ar[dr]^-{\pi^A_x}
                &
                (\ker T_xJ^A, \omega_{J^A}(x))
                \ar[r]^-{i^A_x}
                \ar[d]^-{pr^{J^A}}
                &
                (T_xM,\omega^A(x))
                \ar[d]
                \\
                \left(\displaystyle\frac{T_x(J^{-1}(\mu))}{T_x(G_\mu\cdot x)},\omega^A_\mu(\pi_\mu(x))\right)
                \ar[r]
                \ar[rd]^-{\widetilde{\pi}^A_x}
                &
                \left(\displaystyle\frac{\ker T_xJ^A}{\ker \omega^A(x)}, \widetilde{\omega_{J^A}(x)}\right)
                \ar[d]^-{\widetilde{pr_A}}
                \ar[r]^-{\widetilde{i}^A_x}
                &
                \left(\displaystyle\frac{T_xM}{\ker\omega^A(x)},\widetilde{\omega^A(x)}\right)
                \\
                &
                \left(\ds\frac{\left(\ds\frac{\ker\,  T_xJ^A  }{\ker\, \omega^A(x)}\right)}{\{ [\xi_M(x)]/\xi\in \mathfrak{g}_{\mu_A} \} }, \omega_{\mu_A}(x)\right)
                &
                }
            \]

            \end{enumerate}
        \end{proof}

        \begin{prop}\label{step2}
            For $A=1,\ldots, k$, let $\widetilde{\pi}^A_x$ be the linear maps
            defined in Proposition \ref{tildepi}. If every
            $\widetilde{\pi}^A_x$ is an epimorphism and
$\displaystyle\bigcap_{A=1}^k\ker\,\widetilde{\pi}^A_x=\{0\}$, for every $x\in J^{-1}(\mu)$,
then $(\omega^1_\mu,\ldots,\omega^k_\mu)$ is a polysymplectic
            structure on $J^{-1}(\mu)/G_\mu$, which satisfies $\pi_\mu^*\omega^A_\mu = i^*\omega^A$
for every $A$.
        \end{prop}
        \begin{proof} It is a consequence of Lemmas \ref{algebra} and
                \ref{omegamua}.        \end{proof}

        Observe that, after these two steps, a polysymplectic structure is obtained
         on the quotient space $J^{-1}(\mu)/\mathcal{F}_{J^{-1}(\mu)}\equiv J^{-1}(\mu)/G_\mu$
          using a family of auxiliary maps $\widetilde{\pi}^A_x,\; A=1,\ldots, k,\, x\in J^{-1}(\mu)$.
          Nevertheless we want find conditions on the momentum map, and the kernel of the $2$-forms $\omega^A$ such
           that the two conditions of the last proposition on the linear map $\widetilde{\pi}^A_x$ hold.
    \begin{lemma}
        The linear map
        \[
            \widetilde{\pi}^A_x\colon \displaystyle\frac{T_x(J^{-1}(\mu))}{T_x(G_\mu\cdot x)} \to \displaystyle\frac{\left( \ds\frac{\ker (T_xJ^A)}{\ker\omega^A(x)}\right)}{\{[\xi_M(x)] \,\vert\;\xi\in \mathfrak{g}_{\mu_A}\}}
        \]
        is an epimorphism if and only if
        \[
            \ker (T_xJ^A) = T_x(J^{-1}(\mu)) + \ker\omega^A(x)+ T_x(G_{\mu_A}\cdot x)\,.
        \]
    \end{lemma}

    \begin{proof}We remark that

    \begin{equation}\label{cond_1}T_x(J^{-1}(\mu)) = \bigcap_{B=1}^k\ker (T_xJ^B),\quad T_x(G_\mu\cdot x) = \{\xi_M(x) \vert\,\xi\in \mathfrak{g}_{\mu_B},\forall B\}=\bigcap_{B=1}^kT_x(G_{\mu_B}\cdot x)\end{equation} and that
    \begin{equation}\label{cond_2}
        \displaystyle\frac{\left(\displaystyle\frac{\ker (T_xJ^A)}{\ker\omega^A(x)}\right)}{\{[\xi_M(x)] \vert\,\xi\in \mathfrak{g}_{\mu_A}\}} =\displaystyle\frac{\left(\displaystyle\frac{\ker (T_xJ^A)}{\ker\omega^A(x)} \right)}{\left( \displaystyle\frac{T_x(G_{\mu_A}\cdot x)}{T_x(G_{\mu_A}\cdot x) \cap \ker \omega^A(x)}\right)}\,.
    \end{equation}

    Thus, a direct algebraic computation proves the result.

    \end{proof}

    \begin{lemma}
    The condition $\bigcap_{B=1}^k \ker \tilde{\pi}^B_x=\{0\}$ holds if and only if
    \[
        T_x(G_\mu\cdot x) = \bigcap_{B=1}^k \left( \ker \omega^B(x) + T_x(G_{\mu_B}\cdot x)\right) \cap T_x(J^{-1}(\mu))\,.
    \]
    \end{lemma}
    \begin{proof}
    From (\ref{cond_1}), we have that
    \[
        T_x(G_\mu\cdot x) \subseteq \bigcap_{B=1}^k \left(\ker\omega^B(x) + T_x (G_{\mu_B}\cdot x)\right)\,.
    \]

    Therefore,
    \[
        T_x(G_\mu\cdot x) \subseteq \bigcap_{B=1}^k \left(\ker\omega^B(x) + T_x (G_{\mu_B}\cdot x)\right)\cap T_x(J^{-1}(\mu))\,.
    \]

    On the other hand, using (\ref{cond_1}) and (\ref{cond_2}), we deduce that the condition $\bigcap_{B=1}^k\ker \tilde{\pi}_x^B=\{0\}$ holds if and only if
    \[
        \bigcap_{B=1}^k \left(\ker \omega^B(x) + T_x (G_{\mu_B}\cdot x)\right) \cap T_x(J^{-1}(\mu))\subseteq T_x(G_\mu\cdot x)\,.
    \]
    This proves the result.
    \end{proof}

        Finally, we can summarize the  results of this section in the following reduction theorem for
polysymplectic manifolds.

        \begin{theorem}\label{MW poly th}
            Let $(M,\omega^1,\ldots, \omega^k; \Phi, J)$ be a polysymplectic Hamiltonian $G$-space such that
     $\mu=(\mu_1,\ldots,$ $ \mu_k)\in  \mathfrak{g}^* {\bf \times} \stackrel{k}{\dots} {\bf \times} \mathfrak{g}^*$
                is a regular value of $J$ and
                $G_\mu$ acts freely and properly on $J^{-1}(\mu)$.
            Assume that for every $x\in J^{-1}(\mu)$ the following conditions hold:

                \begin{equation}\label{MW_cond_1}
                    \ker (T_xJ^A) = T_x(J^{-1}(\mu)) + \ker\omega^A(x)+ T_x(G_{\mu_A}\cdot x), \, \text{for every $A$},
                \end{equation}
                and
                \begin{equation}\label{MW_cond_2}
                     T_x(G_\mu\cdot x) = \bigcap_{B=1}^k \left( \ker \omega^B(x) + T_x(G_{\mu_B}\cdot x)\right) \cap T_x(J^{-1}(\mu))\,.
                \end{equation}
            Then the orbit space $J^{-1}(\mu)/G_{\mu}$ is a smooth manifold which admits a unique
polysymplectic structure
            $(\omega^1_\mu, \ldots ,\omega^k_\mu)$ satisfying the property
            \begin{equation}\label{redform}
                \pi_\mu^*\omega_\mu^A=i^*\omega^A\, ,
            \end{equation}
            where $\pi_\mu\colon J^{-1}(\mu)\to J^{-1}(\mu)/G_{\mu}$ is the canonical projection,
 and $i\colon J^{-1}(\mu)\to M$ is the canonical inclusion.
        \end{theorem}

\newpage
    \subsection{Examples}
    \subsubsection{The product of symplectic manifolds}\label{ex0}

        Let $M$ be the product of the symplectic manifolds $M_A$, with $A\in \{1,\ldots, k\}$. Denote by $\tilde{\omega}^A$ the symplectic $2$-form on $M_A$ and by $(\omega^1,\ldots, \omega^k)$ the corresponding $k$-polysymplectic structure on $M$ (see Appendix \ref{examplepoly}).

        Suppose that for every $A\in \{1,\ldots, k\}$, $\tilde{\Phi}^A\colon G_A\times M_A \to M_A$ is a free proper symplectic action of the Lie group $G_A$ on $M_A$ which admits a $Coad$-equivariant momentum map $\tilde{J}^A\colon M_A\to \mathfrak{g}^*_A$.

        Then, we may consider the free and proper action $\Phi$ of the product Lie group $G=G_1\times \cdots \times G_k$ on $M$ given by
        \[
            \Phi((g_1,\ldots, g_k), (x_1,\ldots, x_k))= (\tilde{\Phi}^1(g_1,x_1), \ldots, \tilde{\Phi}^k(g_k,x_k))\,.
        \]

        It is clear that $\Phi$ is polysymplectic.

        Moreover, if $\mathfrak{g}=\mathfrak{g}_1\times \ldots \times \mathfrak{g}_k$ is the Lie algebra of $G$, we have that $J=(J^1,\ldots, J^k)$ is a $Coad^k$- equivariant momentum map for the action $\Phi$, where $J^A\colon M\to \mathfrak{g}^*=\mathfrak{g}_1^*\times \ldots\times \mathfrak{g}^*_k$ is defined by
        \[
            J^A(x_1,\ldots, x_k)(\xi_1,\ldots, \xi_k)=\tilde{J}^A(x_A)(\xi_A)\,,
        \]
        for $(x_1,\ldots, x_k)\in M$ and $(\xi_1,\ldots, \xi_k)\in \mathfrak{g}_1\times \ldots \times \mathfrak{g}_k=\mathfrak{g}$.

        Now, assume that $\tilde{\mu}_A\in \mathfrak{g}^*_A$ is a regular value of the momentum map $\tilde{J}^A\colon M_A\to  \mathfrak{g}^*_A$, for $A\in \{1,\ldots,k\}.$

        Then, $\mu=(\mu_1,\ldots, \mu_k)$ is a regular value of the momentum map $J$, with $\mu_A=(0,\ldots, 0,\tilde{\mu}_A, 0,\ldots, 0) \in \mathfrak{g}^*_1\times \cdots\times \mathfrak{g}^*_k=\mathfrak{g}^*$.

        In addition, if $x=(x_1,\ldots, x_k)\in J^{-1}(\mu)$, it follows that
        \begin{align*}
           & \ker (T_xJ^A)  = T_{x_1}M_1\times \cdots \times T_{x_{A-1}}M_{A-1}\times \ker (T_{x_A}\widetilde{J}^A)\times T_{x_{A+1}} M_{A+1}\times \cdots\times T_{x_k}M_k\,,\\
           &T_x(J^{-1}(\mu)) = \ker (T_{x_1}\tilde{J}^1) \times \cdots\times \ker (T_{x_k}\tilde{J}^k)\,,\\
           &
           \ker\omega^A(x) = T_{x_1}M_1\times \cdots\times T_{x_{A-1}}M_{A-1} \times \{0\}\times T_{x_{A+1}}M_{A+1}\times \cdots\times T_{x_k}M_k\,,\\
           &
           T_x(G_{\mu_A}\cdot x) =T_{x_1}(G_1\cdot x_1)\times\cdots \times T_{x_{A-1}}(G_{A-1}\cdot x_{A-1})\times T_{x_A}((G_A)_{\tilde{\mu}_A}\cdot x_A)\times T_{x_{A+1}}(G_{A+1}\cdot x_{A+1})\\
           &\qquad\qquad\qquad\times \cdots\times T_{x_{k}}(G_{k}\cdot x_{k})\,,\\
           &
           T_x(G_\mu\cdot x) = T_{x_1}((G_1)_{\tilde{\mu}_1}\cdot x_1)\times \cdots\times T_{x_k}((G_k)_{\tilde{\mu}_k}\cdot x_k)\,.
        \end{align*}

        Thus, we deduce that
        \[
            \ker (T_xJ^A) = T_{x}(J^{-1}(\mu)) + \ker\omega^A(x),\quad \forall A
        \]
        and
        \[
        \bigcap_{B=1}^k\left(\ker\omega^B(x)+T_x(G_{\mu_B}\cdot x)\right) = T_x(G_\mu\cdot x)\,.
        \]

        Therefore, we may apply Theorem \ref{MW poly th} and we obtain a reduced polysymplectic manifold $J^{-1}(\mu)/G_{\mu}$.

        Note that
        \[
            J^{-1}(\mu) = (\tilde{J}^1)^{-1}(\tilde{\mu}_1) \times \cdots \times (\tilde{J}^k)^{-1}(\tilde{\mu}_k)
        \]
        and
        \[
        G_\mu= (G_1)_{\tilde{\mu}_1}\times \cdots \times (G_k)_{\tilde{\mu}_k}\,.
        \]

        In fact, $J^{-1}(\mu)/G_{\mu}$ is polysymplectomorphic to the product of the reduced symplectic manifolds $\displaystyle\frac{(\tilde{J}^A)^{-1}(\tilde{\mu}_A)}{(G_A)_{\tilde{\mu}_A}}$, with $A\in \{1,\ldots, k\}$, that is,
        \[
            J^{-1}(\mu)/G_{\mu}\cong \displaystyle\frac{(\tilde{J}^1)^{-1}(\tilde{\mu}_1)}{(G_1)_{\tilde{\mu}_1}}\times \cdots \times \displaystyle\frac{(\tilde{J}^k)^{-1}(\tilde{\mu}_k)}{(G_k)_{\tilde{\mu}_k}}\,.
        \]
\subsubsection{The cotangent bundle of $k^1$-covelocities}\label{ex1}

        In this case we consider the model of polysymplectic manifold $M=(T^1_k)^* Q$ (see Appendix A).

        Let $\varphi\colon Q\to Q$ be a diffeomorphism. The {\rm
            canonical prolongation} of $\varphi$ to the bundle of
            $k^1$-covelocities of $Q$, is the map $(T^1_k)^* \varphi\colon (T^1_k)^*Q\to (T^1_k)^* Q$ given by
                \[
                    (T^1_k)^*\varphi(\alpha^1_q,\ldots,\alpha^k_q)=
                    [(\alpha^1_q\circ T\varphi)(\varphi^{-1}(q)),\ldots,(\alpha^k_q\circ T\varphi)(\varphi^{-1}(q))] \ .
                \]
            An interesting property of this map $(T^1_k)^*\varphi$
is that it conserves the canonical polysymplectic structure of
            $(T^1_k)^* Q$, that is,
                \[
                    [(T^1_k)^*\varphi]^*\omega^A=\omega^A \ .
                \]
            Observe that in the case $k=1$, this notion
            reduces to the canonical prolongation $T^*\varphi$ from $Q$ to $T^*Q$.

            Using the canonical prolongation, we can define a polysymplectic action in the following way.

        Every action $\phi\colon G\times Q\to Q$ of a Lie group $G$ on an arbitrary manifold $Q$
can be lifted to a polysymplectic action
                \begin{equation}\label{lifted action}
                \begin{array}{rccl}
                    \Phi=\phi^{T^*_k}\colon & G\times (T^1_k)^*Q & \to & (T^1_k)^* Q\\\noalign{\medskip}
                     & (g,\alpha^1_q,\ldots,\alpha^k_q) &\mapsto &\Phi^{T^*_k}(g,\alpha^1_q,\ldots,\alpha^k_q) =
                     (T^1_k)^*(\Phi_{g^{-1}})(\alpha^1_q,\ldots,\alpha^k_q)\,.
                \end{array}
                \end{equation}

        Now, in order to define a $Coad^k$-equivariant momentum map for this action $\phi^{T^*_k}$,
we recall the following theorem, which can be found in \cite{Gunther-1987,{MRS-2004}}.

        \begin{theorem}\label{momemtum covelocities}
            Let $\Phi\colon G\times M\to M$ be a polysymplectic
            action on a polysymplectic manifold $(M,\omega^1,\ldots,$ $ \omega^k)$.
Assume the polysymplectic structure is
            exact, that is, there exists a family of $1$-forms
            $\theta^1,\dots, \theta^k$ such that,
            $\omega^A=-d\theta^A$. Assume
            that  the action leaves each $\theta^A$ invariant, i.e.,
            $(\Phi_g)^*\theta^A=\theta^A$ for every $g\in G$
            (and then it is called a {\rm $k$-polysymplectic exact action}). Then the mapping
            $J\equiv (J^1,\ldots, J^k)\colon M\to
            \mathfrak{g}^*\times\stackrel{k}{\ldots}\times\mathfrak{g}^*$
            defined by
                \[
                    J^A(x)(\xi)=\theta^A(x)\left(\xi_M(x)\right)\quad ; \quad \xi\in \mathfrak{g}\ ,\ x\in M
                \]
             is a $Coad^k$-equivariant momentum map for $\Phi$.
        \end{theorem}
        \proof It is equivalent to Proposition 6.9 in G\"unther's paper \cite{Gunther-1987}.\qed

        Consider now the special case when $M=(T^1_k)^* Q$ with
        $\theta^1,\ldots, \theta^k$ the canonical $1$-forms. As we
        have seen, a diffeomorphism $\varphi$
        of $Q$ to $Q$ lifts to a diffeomorphism $(T^1_k)^*\varphi$
        of $(T^1_k)^*Q$ that preserves each $\theta^A$, and an action
        $\phi$ of $G$ on $Q$ can be lifted to obtain an action on
        $(T^1_k)^* Q$ (see example \ref{lifted action}).

        \begin{corollary}\label{k-cotangent lift}
            Let $\phi\colon G\times Q\to Q$ be an action of $G$ on
            $Q$ and let $\Phi=\phi^{T^*_k}$ be the lifted action on
            $M=(T^1_k)^*Q$. Then this polysymplectic action has a
            $Coad^k$-equivariant momentum mapping $J\equiv
            (J^1,\ldots,J^k)\colon (T^1_k)^*Q\to
            \mathfrak{g}^*\times\stackrel{k}{\ldots}\times\mathfrak{g}^*$
            given by
                \[
                    J^A(\alpha^1_q,\ldots,\alpha^k_q)(\xi)=\alpha^A_q(\xi_Q(q))
                \]
            where $\xi_Q$ is the infinitesimal generator of $\phi$
            on $Q$.
        \end{corollary}

        We consider the Hamiltonian polysymplectic $G$-space $(M,\omega^1,\ldots,\omega^k; \Phi;J)$ where
        \begin{itemize}
            \item $M=(T^1_k)^*Q$ is the tangent bundle of $k^1$-covelocities of a manifold $Q$,
            with the canonical polysymplectic structure defined in Appendix A.
     \item The polysymplectic action is $\Phi=\phi^{T^*_k}$, the lift  of an action $\phi\colon G\times Q\to Q$,
           (see  (\ref{lifted action}) ).
            \item The $Coad^k$- equivariant momentum map $J\equiv(J^1,\ldots, J^k)\colon (T^1_k)^*Q
            \to\mathfrak{g}^*\times\stackrel{k}{\ldots}\times\mathfrak{g}^*$
            for the action $\Phi^{T^*_k}$ is defined by (see corollary \ref{k-cotangent lift})
                    \[
                        J^A(\alpha^1_q,\ldots,\alpha^k_q)(\xi)=\alpha^A_q\left(\xi_Q(q)\right)\,,\quad A=1,\ldots, k\,.
                    \]
        \end{itemize}

        If $\xi\in\mathfrak{g}$, we denote by $\xi_Q$ the
        infinitesimal generator of the action $\phi$ associated to
        $\xi$, by $\xi_{T^*Q}$ the infinitesimal generator of the
        cotangent lifting $\phi^{T^*}$ of the action $\phi$
        associated to $\xi$, and finally,  by $\xi_{(T^1_k)^* Q}$ the
        infinitesimal generator of $\phi^{T^*_k}$ associated to
        $\xi$. It is immediate to prove that
            \begin{itemize}
                \item $\xi_{T^*Q}$ is $\pi_Q$-projectable on
                $\xi_Q$.
                \item $\xi_{(T^1_k)^*Q}$ is $\pi^{k,A}_Q$-projectable on
                $\xi_{T^*Q}$, and $\pi^k_Q$-projectable on
                $\xi_{Q}$.
            \end{itemize}

        Next, we will see that if the action $\Phi$ is infinitesimally free then the Hamiltonian polysymplectic $G$-space $(M,\omega^1,\ldots, \omega^k; \Phi; J)$ satisfies the hypotheses of Theorem \ref{MW poly th}. We remark that if $\Phi$ is free then $\Phi$ is infinitesimally free.

        Denote by $\mathcal{J}\colon T^*Q\to \mathfrak{g}^*$ the standard momentum map associated to the action $\phi\colon G\times Q\to Q$, that is, for $\alpha_q\in T^*_qQ$,
        \[
            \begin{array}{rccl}
                \mathcal{J}(\alpha_q)\colon & \mathfrak{g} & \to & \mathbb{R}\\\noalign{\medskip}
                & \xi & \mapsto & \mathcal{J}(\alpha_q)(\xi)=\alpha_q(\xi_Q(q))\,.
            \end{array}
        \]
    \begin{lemma}\label{lemma 320}
    If $\phi$ is infinitesimally free, that is, the linear map
    \[
        \begin{array}{rcl}
            \mathfrak{g} & \to & T_qQ\\\noalign{\medskip}
            \xi & \mapsto & \xi_Q(q)
        \end{array}
    \]
    is injective, for every $q\in Q$, then
    \[
        \ker (T_{\alpha_q}\mathcal{J}) + V_{\alpha_q}(\pi_Q) = T_{\alpha_q} (T^*Q)\,.
    \]
    \end{lemma}
    \begin{proof}
        If $\Phi$ is infinitesimally free then $\mathcal{J}\vert_{T_q^*Q}\colon T_q^* Q\to \mathfrak{g}^*$ is a linear epimorphism. This implies that $\mathcal{J}\colon T^*Q\to \mathfrak{g}^*$ is a submersion and
        \begin{equation}\label{ker_mo}
            \dim \ker (T_{\alpha_q}\mathcal{J}) = 2\dim Q - \dim G\,.
        \end{equation}

        Now, let $^V _{\alpha_q}\colon T^*_qQ\to T_{\alpha_q}(T^*_qQ) = V_{\alpha_q}(\pi_Q)\subseteq T_{\alpha_q}(T^*Q)$ be the canonical isomorphism. A direct computation proves that
        \[
            \ker (T_{\alpha_q}\mathcal{J}) \cap V_{\alpha_q}(\pi_Q) = \{(\beta_q)^V_{\alpha_q}\in V_{\alpha_q}(\pi_Q) \, \vert\, \beta_q\in T^0_q(G\cdot q)\},
        \]
        where $T^0_q(G\cdot q)$ is the annihilator of the subspace $T_q(G\cdot q)$.

        Thus,
        \begin{equation}\label{Inter}
            \dim\left(\ker (T_{\alpha_q}\mathcal{J}) \cap V_{\alpha_q}(\pi_Q)\right) = \dim Q - \dim G\,.
        \end{equation}

        Therefore, using (\ref{ker_mo}) and (\ref{Inter}), we deduce that
        \[
            \dim\left(\ker (T_{\alpha_q}\mathcal{J}) + V_{\alpha_q}(\pi_Q)\right)=2\dim Q=\dim T_{\alpha_q}(T^*Q)
        \]
        which ends the proof of the result.
    \end{proof}

    If $\phi$ is infinitesimally free, then the momentum map $\mathcal{J}\colon T^*Q\to \mathfrak{g}^*$ is a submersion (see proof of Lemma \ref{lemma 320}) and this implies that the polysymplectic momentum map $J\colon (T^1_k)^*Q\to \mathfrak{g}^*\times \stackrel{k)}{\cdots}\times \mathfrak{g}^*$ also is a submersion. Thus, every $\mu\in \mathfrak{g}^*\times \stackrel{k)}{\cdots}\times \mathfrak{g}^*$ is a regular value of $J$. Moreover, we may prove the following results.
    \begin{prop}\label{prop 321}
        If $\phi\colon G\times Q\to Q$ is infinitesimally free, $\alpha_q=(\alpha^1_q,\ldots, \alpha^k_q)\in (T^1_k)^*_qQ$ and $\mu = J(\alpha_q)\in \mathfrak{g}^*\times \stackrel{k)}{\cdots}\times \mathfrak{g}^*$, we have that
        \[
            \ker (T_{\alpha_q} J^A) = T_{\alpha_q}(J^{-1}(\mu)) + \ker \omega^A(\alpha_q),\quad \makebox{ for all $A$}.
        \]
    \end{prop}
    \begin{proof}
        From Corollary \ref{k-cotangent lift} it follows that
        \[
            \ker (T_{\alpha_q} J^A) = T_{\alpha_q}((T^1_k)^* Q) \cap \left(
            T_{\alpha^1_q}(T^*Q)\times \cdots \times T_{\alpha^{A-1}_q}(T^*Q)\times \ker (T_{\alpha_q^A}\mathcal{J})\times T_{\alpha^{A+1}_q}(T^*Q)\times \cdots \times T_{\alpha^k_q}(T^*Q)
            \right)
        \]
        and
        \[
            T_{\alpha_q}(J^{-1}\mu) = T_{\alpha_q} ((T^1_k)^*Q) \cap \left( \ker (T_{\alpha^1_q}\mathcal{J}) \times \cdots \times \ker (T_{\alpha^k_q}\mathcal{J}) \right)\,.
        \]

        On the other hand, using (\ref{locexp}), we deduce that
        \[
            \ker(\omega^A(\alpha_q)) = V_{\alpha^1_q}(\pi_Q) \times \cdots\times V_{\alpha^{A-1}_q}(\pi_Q)\times \{0\} \times V_{\alpha^{A+1}_q}(\pi_Q) \times \cdots \times V_{\alpha^k_q}(\pi_Q)\,.
        \]
        Therefore, from Lemma \ref{lemma 320}, the result follows.
    \end{proof}
    \begin{prop}\label{prop322}
        If $\phi\colon G\times Q\to Q$ is infinitesimally free, $\alpha_q=(\alpha^1_q,\ldots \alpha^k_q)\in (T^1_k)^* Q$ and $\mu = J(\alpha_q)\in \mathfrak{g}^*\times \stackrel{k)}{\cdots}\times \mathfrak{g}^*$ then
        \[
            \bigcap_{B=1}^k \left ( \ker\omega^B(\alpha_q) + T_{\alpha_q}(G_{\mu_B}\cdot \alpha_q)\right) = T_{\alpha_q}(G_\mu\cdot \alpha_q)\,.
        \]
    \end{prop}
    \begin{proof}
        We have that
        \[
            \ker\omega^B(\alpha_q) + T_{\alpha_q}(G_{\mu_B}\cdot \alpha_q) = \{ (v_1^B+\xi_{T^*Q}(\alpha^1_q), \ldots, v^B_k+\xi_{T^*Q}(\alpha^k_q))\,|\, \xi \in \mathfrak{g}_{\mu_B}, v^B_A\in V_{\alpha^A_q}(\pi_Q), \makebox{for every $A$}\}\,.
        \]

        Now, using that the action $\phi$ is infinitesimally free and the fact that $\xi_{T^*Q}$ is a $\pi_Q$-projectable over $\xi_Q$, we deduce that
        \[
            \bigcap_{B=1}^k\left( \ker\omega^B(\alpha_q) + T_{\alpha_q}(G_{\mu_B}\cdot \alpha_q)
            \right) = \{(\xi_{T^*Q}(\alpha^1_q),\ldots, \xi_{T^*Q}(\alpha^k_q))\, \vert\, \xi\in \mathfrak{g}_{\mu_B}, \forall B\}\,.
        \]

        Thus, since $\mathfrak{g}_\mu=\bigcap_{B=1}^k \mathfrak{g}_{\mu_B}$, it follows that
        \[
            T_{\alpha_q} (G_\mu\cdot \alpha_q) = \{(\xi_{T^*Q}(\alpha^1_q), \ldots, \xi_{T^*Q}(\alpha^k_q)) \, \vert\, \xi\in \mathfrak{g}_{\mu_B}, \forall B\}
        \]
        which proves the result.
    \end{proof}

    From Theorem \ref{MW poly th} and Propositions \ref{prop 321} and \ref{prop322}, we conclude that if $\phi$ is an infinitesimally free action, $\mu\in \mathfrak{g}^*\times \stackrel{k)}{\cdots}\times \mathfrak{g}^*$ and $G_\mu$ acts properly on $J^{-1}(\mu)$, then $J^{-1}(\mu) /G_\mu$ is a polysymplectic manifold.

    \subsubsection{Kirillov-Kostant-Souriau theorem for polysymplectic manifolds}
    \label{KKSth}

        In this case, we specialize the above example, taking $Q=G$
        with $G$ acting on itself by left translations, that is,
        $\phi_g\equiv L_g$ for every $g\in G$.

        The momentum mapping of the action $\Phi^{T^*_k}$ is $J\equiv(J^1,\ldots, J^k)\colon
        (T^1_k)^*G\to\mathfrak{g}^*\times\stackrel{k}{\ldots}\times\mathfrak{g}^*$
        where
            \[
                J^A(\alpha^1_g,\ldots,\alpha^k_g)(\xi)=\alpha^A_g\circ T_eR_g(\xi)\,,\quad \xi\in\mathfrak{g} \ ,
            \]
        where $R_g$ denotes the right translation by $g\in G$.

        Using the identification
            \[
                \begin{array}{rcl}
                    (T^1_k)^*G\equiv T^*G\oplus\stackrel{k}{\dots}\oplus T^*G & \cong &
                    G\times (\mathfrak{g}^*\times\stackrel{k}{\ldots}\times \mathfrak{g}^*)\\\noalign{\medskip}
      (\alpha^1_g,\ldots, \alpha^k_g) & \equiv & (g,\alpha^1_g\circ T_eL_g,\ldots,\alpha^k_g\circ T_eL_g)
                \end{array}
            \]
        the momentum mapping $J$ can be written as follows:
            \[
                \begin{array}{rccl}
                    J\colon &   G\times (\mathfrak{g}^*\times\stackrel{k}{\ldots}\times \mathfrak{g}^*) &
                    \to  &\mathfrak{g}^*\times\stackrel{k}{\ldots}\times \mathfrak{g}^*\\\noalign{\medskip}
                    &(g,\nu_1,\ldots\nu_k) &\mapsto & (Coad_g(\nu_1),\ldots,
Coad_g(\nu_k))=Coad^k_g(\nu_1,\ldots\nu_k)\,.
                \end{array}
            \]

        On the other hand, it is well-known that if $\omega$ is the canonical symplectic structure of $T^*G$
 then, under the identification $T^*G\cong G\times \mathfrak{g}^*$, we have that
        \[
            \omega(g,\nu)(((T_eLg)(\xi), \alpha), ((T_eLg)(\eta),\beta)) = -\alpha(\eta) + \beta(\xi) + \nu[\xi,\eta]\,,
        \]
        for $(g,\nu)\in G\times \mathfrak{g}^*,\, \xi,\eta\in \mathfrak{g}$ and
$\alpha,\beta\in \mathfrak{g}^*$ (see, for instance, \cite{AM-1978}).

        Thus, if $(\omega^1,\ldots, \omega^k)$ is the polysymplectic structure on
$(T^1_k)^*G\cong G\times \mathfrak{g}^*\times\stackrel{k}{\ldots}\times \mathfrak{g}^*$ it follows that
        \begin{equation}\label{omegaA_kksth}
            \omega^A(g,\nu_1,\ldots, \nu_k)(((T_eLg)(\xi), \alpha_1,\ldots, \alpha_k), ((T_eLg)(\eta),
\beta_1,\ldots, \beta_k)) = -\alpha_A(\eta) + \beta_A(\xi) + \nu_A[\xi,\eta]\,,
        \end{equation}
        for $g,\in G,\, \xi,\eta\in \mathfrak{g}$ and $(\nu_1,\ldots, \nu_k),\,(\alpha_1,\ldots, \alpha_k),
(\beta_1,\ldots, \beta_k)\in \mathfrak{g}^*\times\stackrel{k}{\ldots}\times \mathfrak{g}^*$.

        Then, if $\mu=(\mu_1,\ldots, \mu_k)\in\mathfrak{g}^*\times \stackrel{k}{\ldots}\times \mathfrak{g}^*$
 we have that
            \[
                J^{-1}(\mu_1,\ldots, \mu_k)=
                \{(g,\nu_1,\ldots\nu_k)\in G\times (\mathfrak{g}^*\times\stackrel{k}{\ldots}\times \mathfrak{g}^*) \
\mid \  Coad_g(\nu_A)=\mu_ A\}\ .
                \]
        Therefore, there exists a diffeomorphism between $J^{-1}(\mu)=J^{-1}(\mu_1,\ldots, \mu_k)$
        and $G$ given by
            \[
                \begin{array}{lcl}
                    G & \to &  J^{-1}(\mu_1,\ldots, \mu_k)\\\noalign{\medskip}
                    g & \to & (g,Coad_{g^{-1}}(\mu_1),\ldots, Coad_{g^{-1}}(\mu_k)) \ .
                \end{array}
            \]
        Thus,
            \[
                J^{-1}(\mu_1,\ldots,\mu_k)/G_{(\mu_1,\ldots,\mu_k)}\cong G/G_{(\mu_1,\ldots,\mu_k)} \cong
                \mathcal{O}_{(\mu_1,\ldots,\mu_k)}\subseteq \mathfrak{g}^*\times\stackrel{k}{\ldots}\times
\mathfrak{g}^*\,,
            \]
        that is, the ``reduced phase space'' is just the orbit of
        the $k$-coadjoint action at $\mu=(\mu_1,\ldots,\mu_k)$.
        As a consequence, as the action of $G$ on itself is free, using the results from Section \ref{ex1}
        we deduce that $\mathcal{O}_{(\mu_1,\ldots,\mu_k)}$ is a polysymplectic manifold.

        \begin{remark} {\rm
            In the case $k=1$, this result reduces to the following: \textit{the orbit of
            $\mu\in\mathfrak{g}^*$ under the coadjoint
            representation is a symplectic manifold}. This is the
            statement of  Kirillov-Kostant-Souriau theorem (
            see, for instance, \cite{AM-1978, MW-1974}).}
        \end{remark}


        Note that, under the previous identifications, the canonical projection
$\pi_\mu\colon J^{-1}(\mu)\to J^{-1}(\mu)/G_\mu$ is just the map $\pi_\mu\colon G\to \mathcal{O}_\mu$
given by $$\pi_\mu(g)=Coad^k_{g^{-1}}\mu$$ and
        \[
            T_g\pi_\mu ((T_eL_g)(\xi))=-\xi_{\mathfrak{g}^*\times\stackrel{k}{\ldots}\times\mathfrak{g}^*}
(Coad^k_{g^{-1}}\mu)
        \]
        for $g\in G$ and $\xi\in \mathfrak{g}$.

        Consequently, using (\ref{omegaA_kksth}) and the fact that
$\pi_\mu^*\omega^A_\mu=i^*\omega^A$, it follows that
            \begin{equation}\label{kkspolystruc}
                \omega^A_\mu(\nu)\left(\xi_{\mathfrak{g}^*\times\stackrel{k}{\ldots}\times
        \mathfrak{g}^*}(\nu), \eta _{\mathfrak{g}^*\times\stackrel{k}{\ldots}\times
        \mathfrak{g}^*}(\nu)\right) = -\nu_A[\xi,\eta]\,,
            \end{equation}
            for $\nu\in \mathcal{O}_\mu$ and $\xi,\eta\in \mathfrak{g}$.

        Observe that this polysymplectic structure coincides
        with the polysymplectic structure on
        $\mathcal{O}_{(\mu_1,\ldots,\mu_k)}$ described in (\ref{polysymp-orbit}).

%
%
        Now we consider the Kirillov-Kostant-Souriau theorem for the special case when
        $G=SO(3)$ (the rotation group), and we calculate the reduced polysymplectic structure.
        First, we briefly recall
        the main formulas regarding the special orthogonal group $SO(3)$,
        its Lie algebra $\mathfrak{so}(3)$, and its dual $\mathfrak{so}(3)^*$
        (for more details see, for instance, \cite{RTSST-2005})

        The Lie algebra $\mathfrak{so}(3)$ of $SO(3)$ can be identified with
        $\mathbb{R}^3$ as follows: we define the vector space isomorphism
        $\hat{\;}\colon\mathbb{R}^3\to \mathfrak{so}(3)$, by
$$
            \mathbf{x}=(x_1,x_2,x_3) \mapsto \hat{\mathbf{x}}=\left(
                                                                  \begin{array}{ccc}
                                                                    0 & -x_3 & x_2 \\
                                                                    x_3 & 0 & -x_1 \\
                                                                    -x_2 & x_1 & 0 \\
                                                                  \end{array}
                                                                \right).
$$
        As $(\mathbf{x}\times\mathbf{y})
        \hat{\;}=[\hat{\mathbf{x}},\hat{\mathbf{y}}]$, the map $\hat{\;}$ is
        a Lie algebra isomorphism between $\mathbb{R}^3$, with the cross
        product, and $(\mathfrak{so}(3),[\cdot,\cdot])$, where
        $[\cdot,\cdot]$ is the commutator  of matrices.

        Note that the identity
            \[
               \hat{\mathbf{x}}\mathbf{y}=\mathbf{x}\times\mathbf{y}\makebox{ for every }
               \mathbf{x},\mathbf{y}\in\mathbb{R}^3
            \]
        characterizes this isomorphism. We also note that the standard ``dot''
        product may be written as
            \[
                \mathbf{x}\cdot\mathbf{y}=\frac{1}
                {2}trace(\hat{\mathbf{x}}^{T}\hat{\mathbf{y}})=-\frac{1}
                {2}trace(\hat{\mathbf{x}}\hat{\mathbf{y}}) \ .
            \]

         It is well known that the adjoint representation
         $Ad\colon SO(3)\to  Aut(\mathfrak{so}(3))$ is given by
$$
            Ad_A\hat{\mathbf{x}}= A\hat{\mathbf{x}}A^T = (A\mathbf{x}) \hat{\;} \ ,
$$
    for every $A\in SO(3)$ and $\hat{\mathbf{x}}\in \mathfrak{so}(3)$.
         Using the isomorphism $\hat{\;}$, this action can be regarded as the
         action of $SO(3)$ on $\mathbb{R}^3$, given by
         $Ad_A\mathbf{x}=A\mathbf{x}$.

        The dual $\mathfrak{so}(3)^*$ is identified with
        $(\mathbb{R}^3,\times)$  by the isomorphism
        $\bar{\;}\colon\mathbb{R}^3\to \mathfrak{so}(3)^*$ given by
        $\bar{\mathbf{x}}(\hat{\mathbf{y}})\colon =
        \mathbf{x}\cdot\mathbf{y}$ for every
        $\mathbf{x},\mathbf{y}\in\mathbb{R}^3$. Then the coadjoint action of
        $SO(3)$ on $\mathfrak{so}(3)$ is given by
$$
            Coad(A,\bar{\mathbf{x}})=Ad_{A^{-1}}^*\bar{\mathbf{x}} = (A\mathbf{x}) \bar{\;}\,.
$$

        It is well known that the \textit{coadjoint orbit} associated to $SO(3)$ at
        $\pi_0\in\mathbb{R}^3\equiv\mathfrak{so}(3)^*$ ($\pi_0\neq (0,0,0)$)
is the $2$-sphere $S^2(||\pi_0||)$ and it has a symplectic
        structure given by
            \begin{equation}\label{symplorbit}
                \omega_{\pi_o}(\pi)(\xi,\eta)=-\pi\cdot (\xi\times\eta) \ ,
            \end{equation}
        where $\pi\in \mathcal{O}_{\pi_0}\equiv S^2(||\pi_0||)$, and
        $\xi,\eta\in T_\pi \mathcal{O}_{\pi_0}
        =\{\mathbf{v}\in\mathbb{R}^3\equiv T_{\pi}\mathbb{R}^3 \,\vert\,
        \mathbf{v}\in T_\pi S^2(||\pi_0||)\}$.

        Now we describe the $2$-coadjoint orbit at
        $\mu=(\mu^0_1,\mu^0_2)\in \mathfrak{so}(3)^*\times
        \mathfrak{so}(3)^*$.
        Using the above identifications, the $2$-coadjoint action
        $Coad^2\colon SO(3)\times \mathfrak{so}(3)^*\times
        \mathfrak{so}(3)^* \to \mathfrak{so}(3)^*\times \mathfrak{so}(3)^*$
        can be identified with the natural action
        \[\begin{array}{rccl}
         Coad^2\colon & SO(3)\times \mathbb{R}^3\times\mathbb{R}^3 & \to &
          \mathbb{R}^3\times\mathbb{R}^3\\\noalign{\medskip}
          & (A, \mathbf{\pi}_1,\mathbf{\pi}_2) & \mapsto & (A\mathbf{\pi}_1,A\mathbf{\pi}_2)
        \end{array}  \ .
        \]
        Then,  the
        $2$-coadjoint orbit $SO(3)\cdot (\mathbf{\pi}_1^0,\mathbf{\pi}_2^0)$
        at $ (\mathbf{\pi}_1^0,\mathbf{\pi}_2^0)\in
        \mathbb{R}^3\times\mathbb{R}^3$ is
        \[
            \mathcal{O}_{(\mathbf{\pi}_1^0,
            \mathbf{\pi}_2^0)}=\{(A\mathbf{\pi}_1^0,A\mathbf{\pi}_2^0)
            \in \mathbb{R}^3\times\mathbb{R}^3
             \ \mid \  A\in SO(3)\}\,.
        \]

        We distinguish the following cases:

        \begin{enumerate}
            \item \textit{The trivial case: $(\mathbf{\pi}_1^0,\mathbf{\pi}_2^0)= (0,0)$.}

                 In this case it is immediate that
                 $\mathcal{O}_{(\mathbf{\pi}_1^0,\mathbf{\pi}_2^0)}=0$.
            \item \textit{$\mathbf{\pi}_1^0 $ and $\mathbf{\pi}_2^0$ are
            linearly dependent and $(\mathbf{\pi}_1^0,\mathbf{\pi}_2^0)\neq (0,0)$.}

                Assume that $\mathbf{\pi}^0_1\neq 0$ and
                $\mathbf{\pi}^0_2= \lambda_0\mathbf{\pi}^0_1$,
                with $\lambda_0\in \mathbb{R}$. Then,
                \[
                    \begin{array}{lcl}
                        \mathcal{O}_{(\mathbf{\pi}_1^0,\mathbf{\pi}_2^0)}
                        &=& \{(A\mathbf{\pi}^0_1,\lambda_0 A\mathbf{\pi}^0_1)
                        \in \mathbb{R}^3\times   \mathbb{R}^3 \ \mid \  A\in SO(3)\}
                         \\\noalign{\medskip}&=&
                         \{(\pi,\lambda_0\pi)\in \mathbb{R}^3\times   \mathbb{R}^3 \ \mid \
                         \pi\in S^2(||\mathbf{\pi}^0_1||)\}
                         \\\noalign{\medskip}&\cong & \{\pi\in\mathbb{R}^3 \ \mid \
                         \pi\in S^2(||\mathbf{\pi}^0_1||)\} = S^2(||\mathbf{\pi}^0_1||)\,.
                    \end{array}
                \]
                We know that the orbit
                $\mathcal{O}_{(\mathbf{\pi}_1^0,\mathbf{\pi}_2^0)}$
                (and therefore $S^2(||\mathbf{\pi}^0_1||)$) is a
                polysymplectic manifold. Then, let $\pi\in S^2(||\mathbf{\pi}^0_1||)
                \equiv\mathcal{O}_{(\mathbf{\pi}_1^0,\mathbf{\pi}_2^0)}$; therefore
                    \[
                        T_{(\pi,\lambda_0\pi)}\mathcal{O}_{(\mathbf{\pi}_1^0,\mathbf{\pi}_2^0)} =
                        \{(\mathbf{v},\lambda_0\mathbf{v})\in \mathbb{R}^3\times \mathbb{R}^3 \equiv
                        T_\pi\mathbb{R}^3\times T_{\lambda_0\pi}\mathbb{R}^3 \ \mid \
                        \mathbf{v}\in T_{\pi}S^2(||\pi^0_1||)\} \ .
                    \]
                From (\ref{kkspolystruc}) and (\ref{symplorbit})
                we obtain the polysymplectic structure of
                $\mathcal{O}_{(\mathbf{\pi}_1^0,\mathbf{\pi}_2^0)}$:
                for $\pi\in S^2(||\mathbf{\pi}^0_1||),
                \mathbf{u},\mathbf{v}\in
                T_{\pi}S^2(||\pi^0_1||)$, this polysymplectic
                structure is given by
                \[
                    \begin{array}{lcl}
                        \omega^1_{(\pi^0_1,\pi^0_2)}(\pi,\lambda_0\pi)((\mathbf{u},
                        \lambda_0\mathbf{u}),(\mathbf{v},\lambda_0\mathbf{v})) &=&
                         -\pi\cdot (\mathbf{u}\times\mathbf{v})\\\noalign{\medskip}
                        \omega^2_{(\pi^0_1,\pi^0_2)}(\pi,\lambda_0\pi)((\mathbf{u},
                        \lambda_0\mathbf{u}),(\mathbf{v},\lambda_0\mathbf{v})) &=&
                         -\lambda_0\,\pi\cdot (\mathbf{u}\times\mathbf{v}) \ .
                    \end{array}
                \]
                Thus, under the canonical identification between
                $\mathcal{O}_{(\mathbf{\pi}_1^0,\mathbf{\pi}_2^0)}$
                and $S^2(||\mathbf{\pi}^0_1||)$, the
                $2$-polysymplectic structure on
                $S^2(||\mathbf{\pi}^0_1||)$ is given by
                \[
                    \begin{array}{lcl}
                        \omega^1(\pi)(\mathbf{u},\mathbf{v}) &=& -\pi\cdot
                         (\mathbf{u}\times\mathbf{v})\\\noalign{\medskip}
                        \omega^2(\pi)(\mathbf{u},\mathbf{v})&=& -\lambda_0\,\pi\cdot
                         (\mathbf{u}\times\mathbf{v}) \ .
                    \end{array}
                \]
            \item \textit{$\mathbf{\pi}_1^0 $ and $\mathbf{\pi}_2^0$ are
            linearly independent.}

            In this case there exist a diffeomorphism between
            $\mathcal{O}_{(\mathbf{\pi}_1^0,\mathbf{\pi}_2^0)}$ and
            $SO(3)$ given by the map
$$
                \begin{array}{rccl}
                            Coad^2_{(\mathbf{\pi}_1^0,\mathbf{\pi}_2^0)}\colon &SO(3) & \to &
                            \mathcal{O}_{(\mathbf{\pi}_1^0,\mathbf{\pi}_2^0)}\\\noalign{\medskip}
                             &A & \mapsto &   Coad^2_{(\mathbf{\pi}_1^0,\mathbf{\pi}_2^0)}
                             (A)=(A\pi^0_1,A\pi^0_2)\,.
                \end{array}
$$
            We need only to prove that this map is injective. Assume that
            $A, A'\in SO(3)$ are such that
            $Coad^2_{(\mathbf{\pi}_1^0,\mathbf{\pi}_2^0)}(A)=Coad^2_{(\mathbf{\pi}_1^0,
            \mathbf{\pi}_2^0)}(A')$,
            then for every $i=1,2$, $(A^TA')\pi^0_i=\pi^0_i$.
            Let $U^0=\langle\pi^0_1,\pi^0_2\rangle$ be the $2$-dimensional
            subspace of $\mathbb{R}^3$ generated by $\pi^0_1$ and
            $\pi^0_2$. If $B\colon=A^{-1}A'$, then $B\pi=\pi$ for every             $\pi\in U^0$.
            Now consider an orthonormal basis
            $\{\bar{\pi}^0_1,\bar{\pi}^0_2\}$ of $U^0$ and extend it  to
            a positively oriented orthonormal basis of  $\mathbb{R}^3$;
            that is,
            \[
                \{\bar{\pi}^0_1,\bar{\pi}^0_2,\bar{\pi}^0_3=\bar{\pi}^0_1\times \bar{\pi}^0_2\}\,.
            \]
            As $B\in SO(3)$ and $(U^0)^\bot=<\bar{\pi}^0_3>$, we
            obtain that $B\bar{\pi}^0_3\in (U^0)^\bot$; that is,
            $B\bar{\pi}^0_3=\lambda_0\bar{\pi}^0_3$, but as
            $B\bar{\pi}^0_3$ is unitary and
            $\{B\bar{\pi}^0_1,B\bar{\pi}^0_2,B\bar{\pi}^0_3\}$  must be
            a positively oriented basis, we deduce that $\lambda_0=1$.
            Therefore,
            \[
                B\pi=\pi\quad \forall \pi\in\mathbb{R}^3
            \]
            and so $B=A^{-1}A'=I$, that is $A= A'$.
            Therefore, we can identify
            $\mathcal{O}_{(\mathbf{\pi}_1^0,\mathbf{\pi}_2^0)}$
            with $SO(3)$. We know that
            $\mathcal{O}_{(\mathbf{\pi}_1^0,\mathbf{\pi}_2^0)}$
            is a $2$-polysymplectic manifold, and we will describe
            this structure.

            The diffeomorphism
            $Coad^2_{(\mathbf{\pi}_1^0,\mathbf{\pi}_2^0)}$ is
            equivariant with respect to the action of $SO(3)$ on itself
            by left translations and the action $Coad^2$ of
            $SO(3)$ on
            $\mathcal{O}_{(\mathbf{\pi}_1^0,\mathbf{\pi}_2^0)}$,
            that is, the following condition holds for every
            $A\in SO(3)$,
            \[
                Coad^2_A\circ Coad^2_{(\pi^0_1,\pi^0_2)}=Coad^2_{(\pi^0_1,\pi^0_2)}\circ L_A\,.
            \]

            \begin{lemma}
                The $2$-polysymplectic structure on
                $\mathcal{O}_{(\mathbf{\pi}_1^0,\mathbf{\pi}_2^0)}$
                is invariant by the action $Coad^2$.
            \end{lemma}
            \begin{proof}
                Let $\omega_{\pi^0_i}$ be the symplectic
                structure on $\mathcal{O}_{\pi^0_i},\ i=1,2$.
                This structure is invariant by the action $Coad$
                (see \cite{AM-1978}, pag 485).
	Furthermore,
                    \[
                        \omega^i_{(\pi^0_1,\pi^0_2)}=pr_i^*\omega_{\pi^0_i}\ ,
                    \]
                where $pr_i\colon\mathcal{O}_{(\pi^0_1,\pi^0_2)} \to\mathcal{O}_{\pi^0_i}$
                is the projection (see Proposition \ref{A4}).
                Thus, we obtain:
                \[
                    \begin{array}{lcl}
                         (Coad^2_A)^*\omega^i_{(\pi^0_1,\pi^0_2)} &=&
                         (pr_i\circ Coad^2_A)^*\omega_{\pi^0_i}
                         = (Coad_A\circ pr_i)^*\omega_{\pi^0_i} \\
                         &=& pr_i^*\left((Coad_A)^*\omega_{\pi^0_i}\right)
                         = pr_i^*\omega_{\pi^0_i} = \omega^i_{(\pi^0_1,\pi^0_2)}\,.
                    \end{array}
                \]
            \end{proof}

            As a consequence of the above lemma, we have that
            the $2$-polysymplectic structure
            $(\omega^1,\omega^2)$ induced on $SO(3)$ by the
            diffeomorphism $Coad^2_{(\pi^0_1,\pi^0_2)}$ is
            invariant by left translations. Therefore, it is
            sufficient to compute $\omega^1(Id)$ and
            $\omega^2(Id)$.
            Using (\ref{kkspolystruc}), (\ref{symplorbit}) and the fact that the $2$-polysymplectic structure on
            $SO(3)$ is defined by
            \[
                \omega^A\colon=\left(Coad^2_{(\pi^0_1,\pi^0_2)}\right)^*
                \omega^A_{(\pi^0_1,\pi^0_2)}\,,\; A=1,2\, ,
            \]
            we deduce that
$$
                \begin{array}{lcl}
                    \omega^A(Id)(\hat{\xi}_1,\hat{\xi}_2)=  -\pi^{03}_A\,,\\\noalign{\medskip}
                    \omega^A(Id)(\hat{\xi}_2,\hat{\xi}_3)=  -\pi^{01}_A\,,\\\noalign{\medskip}
                    \omega^A(Id)(\hat{\xi}_3,\hat{\xi}_1)=  -\pi^{02}_A\,,
                \end{array}\quad A=1,2
$$
            where $\pi^0_A=(\pi^{01}_A,\pi^{02}_A,\pi^{03}_A)\in
            \mathbb{R}^3\equiv\mathfrak{so}(3)^*$.

            Finally, let $\{\xi_1,\xi_2,\xi_3\}$ be the
            canonical basis  of
            $\mathfrak{so}(3)\cong\mathbb{R}^3$ and
            $\{\xi^1,\xi^2,\xi^3\}$ the dual basis of
            $\mathfrak{so}(3)^*\cong\mathbb{R}^3$. We denote by
            $\{\theta^1,\theta^2,\theta^3\}$ the basis of left
            invariant $1$-forms on $SO(3)$ given by
            \[
                \theta^i(A)=\left(T_A^*L_{A^{-1}}\xi^i\right)(A);\quad A\in SO(3), \, i=1,2,3\,,
            \]
            then we have  that
$$
         \begin{array}{lcl}
                    \omega^1 &=&
        -\pi^{03}_1\theta^1\wedge\theta^2-\pi^{01}_1\theta^2\wedge\theta^3-\pi^{02}_1
        \theta^3\wedge\theta^1\,,\\
                    \noalign{\medskip}
                    \omega^2 &=&
  -\pi^{03}_2\theta^1\wedge\theta^2-\pi^{01}_2\theta^2\wedge\theta^3-\pi^{02}_2\theta^3\wedge\theta^1\,.
                \end{array}
$$
        \end{enumerate}

\section{Polysymplectic Hamiltonian Systems on the reduced space}
\label{Ham-red}

In this Section we study Hamiltonian systems in the reduced space.
First, a brief description of the dynamics in polysymplectic manifolds is done.

\subsection{Hamiltonian systems on polysymplectic manifolds}

        The dynamics in a polysymplectic manifold $(M,\omega^1,\ldots, \omega^k)$ is introduced
        by giving a Hamiltonian function  $ H \colon M\to \r$. The dynamics
        is given by $k$-vector fields; thus, we first recall
        this notion (see for instance \cite{MRS-2004}), which is a
         natural extension of the notion of a vector field.

        Let $M$ be an arbitrary manifold and $\tau^k_M\colon T^1_kM\to M$
 its tangent bundle of $k^1$-velocities, that is the Whitney sum of $k$
copies of the tangent bundle (for a complete description of this manifold, see for instance \cite{RSV-2007}).

        \begin{definition}
            A $k$-vector field ${\bf X}$ on $M$ is a section $\mathbf{X}\colon M\to T^1_kM$ of $\tau^k_M$.
        \end{definition}

        Since $T^1_kM$ may be canonically identified with the Whitney sum of
$k$ copies of $TM$, we deduce that a $k$-vector field $\mathbf{X}$ defines $k$ vector fields
 $X_1,\ldots, X_k$ on $M$ by projecting $\mathbf{X}$ onto every factor.
From now on, we will identify $\mathbf{X}$ with the $k$-tuple $(X_1,\ldots,
        X_k)$. Throughout this paper we denote by $\mathfrak{X}^k(M)$ the
        set of $k$-vector fields on $M$.

        Now assume that $M$ is a polysymplectic manifold with polysymplectic structure
$(\omega^1,\ldots, \omega^k)$. We define a vector bundle morphism $\flat_\omega$ as follows:
        \[
            \begin{array}{rccl}
                \flat_\omega\colon & T^1_kM & \to &T^*M\\\noalign{\medskip}
                 & (v_1,\ldots, v_k) & \mapsto & \flat_\omega(v_1,\ldots, v_k)=
                 \ds\sum_{A=1}^k\imath_{v_A}\omega^A\,.
            \end{array}
        \]
        The above morphism induces a morphism of $\mathcal{C}^\infty(M)$-modules
        between the corresponding space of sections,
        $\flat_\omega\colon \mathfrak{X}^k(M)\to \Omega^1(M)$.

        \begin{lemma}
            The map $\flat_\omega$ is surjective.
        \end{lemma}
        \begin{proof}
            This result is a particular case of the following
            algebraic assertion: \textit{If $V$ is a vector space with
            a $k$-polysymplectic structure $(\omega^1,\ldots,
            \omega^k)$, then the map
            \[
            \begin{array}{rccl}
                \flat_\omega\colon & V\times\stackrel{k}{\ldots}\times V & \to &V^*\\\noalign{\medskip}
                 & (v_1,\ldots, v_k) & \mapsto & \flat_\omega(v_1,\ldots, v_k)=
                \ds\sum_{A=1}^k\imath_{v_A}\omega^A\,
            \end{array}
        \] is surjective.}

        In fact, we first consider the identification
            \beq
                \begin{array}{rlcl}
                    F\colon &V^*\times \stackrel{k}{\ldots}\times V^* &
                     \cong &(V\times\stackrel{k}{\ldots}\times V)^*\\\noalign{\medskip}
                    &(\alpha^1,\ldots,\alpha^k ) & \mapsto & F(\alpha^1,\ldots,\alpha^k) \ ,
                \end{array}
                \label{ident2}
            \eeq
            where $F(\alpha^1,\ldots,\alpha^k )(v_1,\ldots, v_k)={\rm trace}\big(\alpha^A(v_B)\big)=
\ds\sum_{A=1}^k\alpha^A(v_A)$.
            Now, we consider the map
            \[
            \begin{array}{rccl}
                \sharp_\omega\colon & V & \to &(V\times\stackrel{k}{\ldots}\times V)^*\equiv
                V^*\times\stackrel{k}{\ldots}\times V^*\\\noalign{\medskip}
                 & v & \mapsto & \sharp_\omega(v)=(\imath_v\omega^1,\ldots,\imath_v\omega^k)\ .
            \end{array}
        \]

        As $(\omega^1,\ldots, \omega^k)$ is a polysymplectic
        structure, we have $\ker\,\sharp_\omega=\cap_{A=1}^k
        \ker\,\omega^A=\{0\}$, that is, $\sharp_\omega$ is injective
        and thus the dual map $\sharp_\omega^*$ is surjective.

        Finally, using the identification (\ref{ident2}), it is immediate to prove
        that $\flat_\omega=-\sharp_\omega^*$ and therefore
        $\flat_\omega$ is surjective.
        \end{proof}

        Let $ H\in\mathcal{C}^\infty(M)$ be a function on $M$.
        As $d H\in \Omega^1(M)$ and the map $\flat_\omega$ is surjective,
        then there exists a $k$-vector field $\mathbf{X}^ H=(X_1^ H,\ldots, X_k^ H)$ satisfying
        \begin{equation}\label{poly ham eq}
            \flat_\omega(X_1^ H,\ldots, X_k^ H)=d H\,.
        \end{equation}
        This equation (\ref{poly ham eq}) is called the \textit{Hamiltonian polysymplectic equation}.

        \begin{remark}
        {\rm Observe that the solution to (\ref{poly ham eq}) is not, in general,  unique.}
        \end{remark}

        When we consider standard polysymplecic structures
        (that is, when $M$ has an atlas of canonical charts for $(\omega^1,\ldots, \omega^k)$,
        i.e. charts in which locally $(\omega^1,\ldots, \omega^k)$
 is written as the canonical model, see (\ref{locexp})),
        we obtain the classical local formulation of the Hamilton equations.

\subsection{Reduced polysymplectic Hamiltonian systems}

Now we want to induce Hamiltonian polysymplectic systems on the reduced phase space.

\begin{theorem}\label{reduction_dinamic}
    Under the assumptions of Theorem \ref{MW poly th},
     let $H\colon M\to \mathbb{R}$ be a Hamiltonian function which is invariant
     under the action of $G$. We denote by
     $\mathbf{X}^H=(X^H_1,\ldots, X^H_k)$ the $k$-vector field
     associated with $H$ which is a solution to (\ref{poly ham eq}). Assume that each $X^H_A$ satisfies:
    \begin{itemize}
        \item it is $G$-invariant; that is,
            \begin{equation}\label{vfGinv}
                T(\Phi_g)(X^H_A)=X^H_A\ ,\  {\rm for}\  \ g\in G\ , A=1,\ldots, k\ .
            \end{equation}
        \item The restriction $X^H_A\vert_{J^{-1}(\mu)}$ is tangent to $J^{-1}(\mu)$.
    \end{itemize}

    Then the flows $F^A_t$ of $X^H_A$ leave $J^{-1}(\mu)$
    invariant and commute with the action of $G_\mu$ on $J^{-1}(\mu)$,
    so they induce canonically flows $F^A_{t\mu}$ on $J^{-1}(\mu)/G_\mu$
    satisfying that $\pi_\mu \circ F^A_t=F^A_{t\mu}\circ\pi_\mu$.
    If $Y_A$ is the generator of $F^A_{t\mu}$ then $(Y_1,\ldots, Y_k)$
    is a solution to the Hamiltonian polysymplectic system on
    $J^{-1}(\mu)/G_\mu$ associated with a Hamiltonian function
    $H_\mu\colon J^{-1}(\mu)/G_\mu\to \mathbb{R}$
    satisfying that $H_\mu\circ \pi_\mu=H\circ i$.
     $H_\mu$ is called the {\rm reduced Hamiltonian function}.
\end{theorem}
\begin{proof}
 As $X^H_A\vert_{J^{-1}(\mu)}\in T(J^{-1}(\mu))$,
 the flow $F^A_t$ of $X^H_A$ leaves $J^{-1}(\mu)$ invariant.

 From (\ref{vfGinv}) we deduce that $F^A_t\circ \Phi_g=\Phi_g\circ F^A_t$ for every $g\in G_\mu$.
 So, for every $A=1,\ldots, k$, we get a well-defined flow $F^A_{t\mu}$ on $J^{-1}(\mu)/G_\mu$
 such that $\pi_\mu \circ F^A_t=F^A_{t\mu}\circ\pi_\mu$.
Thus, as $H$ is $G$-invariant, we can define the function
$H_\mu\colon J^{-1}(\mu)/G_\mu\to \mathbb{R}$ by $H_\mu([x])=H(x)$, for every $x\in J^{-1}(\mu)$.

 Denote by $Y_A$ the generator of $F^A_{t\mu}$.
 As $\pi_\mu\circ F^A_t=F^A_{t\mu}\circ \pi_\mu$, we have
 \[
    T\pi_\mu\circ X^H_A= Y_A\circ \pi_\mu\,.
 \]
 Using $i^*\omega^A=\pi_\mu^*\omega^A_\mu$, we obtain
 \[
 \begin{array}{lcl}
 (dH_\mu)([v_x]) &=& i^*dH(v_x)=i^*\left(\ds\sum_{A=1}^k\imath_{X^H_A}\omega^A\right)(v_x)=
 \ds\sum_{A=1}^k i^*\omega^A(x)(X^H_A(x),v_x)\\\noalign{\medskip}
 &=& \ds\sum_{A=1}^k (\pi_\mu^*\omega^A_\mu)(x)(X^H_A(x),v_x) =
 \ds\sum_{A=1}^k\omega^A_\mu ([x])( T_x\pi_\mu (X^A_H(x)), T_x\pi_\mu (v_x))\\\noalign{\medskip}
 &=& \ds\sum_{A=1}^k\omega^A_\mu ([x])(Y_A([x]),[v_x])=
 \ds\sum_{A=1}^k(\imath_{Y_A}\omega^A_\mu)([v_x])\, ;
 \end{array}
 \]
 that is, $(Y_1,\ldots, Y_k)$ is a solution to the polysymplectic Hamiltonian equation
 (\ref{poly ham eq})
 on $J^{-1}(\mu)/G_\mu$ associated with $H_\mu$.
\end{proof}

\subsection{Examples}

    \subsubsection{The $k$-cotangent bundle of a Lie group}\label{general group example}

    In this part we discuss an application of Theorem \ref{reduction_dinamic}.
Let $(G,h)$ be a Lie group with a left-invariant metric $h$ and $\mathfrak{g}$ its Lie algebra.

In this example we consider the following canonical identifications
$TG\cong G\times \mathfrak{g}$ and $T^*G\cong G\times \mathfrak{g}^*$, via the diffeomorphisms
        \[
            \begin{array}{ccl}
                TG & \to & G\times \mathfrak{g}\\\noalign{\medskip}
                v_g & \mapsto & (g, T_gL_{g-1}(v_g))
            \end{array} \quad and \quad
            \begin{array}{ccl}
                T^*G & \to & G\times \mathfrak{g}^*\\\noalign{\medskip}
                \alpha_g & \mapsto & (g, \alpha_g\circ T_eL_g)
            \end{array}\,.
        \]

Hence, in a natural way we consider the identifications
        \[
            (T^1_k)^* G\cong T^*G\oplus\stackrel{k}{\cdots}\oplus T^*G\cong
G\times \mathfrak{g}^*\times\stackrel{k}{\cdots}\times\mathfrak{g}^*
        \]
     (see example \ref{KKSth}) and
        \[
            T\left((T^1_k)^*G\right) \cong \big ( G\times \mathfrak{g}^*\times\stackrel{k}{\cdots}
\times\mathfrak{g}^*\big )\times \big (\mathfrak{g}\times \mathfrak{g}^*\times\stackrel{k}{\cdots}
\times\mathfrak{g}^* \big)\,.
        \]

     Using these identifications, we can write the lift to $(T^1_k)^*G$ of the action of $G$
on itself by left translations, as follows:
        \[
            \begin{array}{rcl}
                G\times (G\times\mathfrak{g}^*\times\stackrel{k}{\ldots}
                \times\mathfrak{g}^*) & \to & G\times \mathfrak{g}^*\times\stackrel{k}{\ldots}\times
\mathfrak{g}^*\\\noalign{\medskip}
                (h,(g,\mu_1,\ldots,\mu_k)) & \mapsto & (hg,\mu_1,\ldots,\mu_k)
            \end{array}
        \]

      In this case, the canonical $k$-polysymplectic structure $(\omega^1_G,\ldots, \omega^k_G)$ on
$(T^1_k)^*G$ is defined as follows:
        \[
            \omega^A_G(g,\mu_1,\ldots,\mu_k)\left( ((T_eL_g)(\xi),\nu_1,\ldots, \nu_k),((T_eL_g)(\eta),
\gamma_1,\ldots, \gamma_k)\right) = \gamma_A(\xi)-\nu_A(\eta)+\mu_A[\xi,\eta]\,,\quad A=1,\ldots, k\,,
        \]
     where $(g,\mu_1,\ldots,\mu_k)\in G\times \mathfrak{g}^*\times\stackrel{k}{\cdots}\times\mathfrak{g}^*
        $ and $((T_eL_g)(\xi),\nu_1,\ldots, \nu_k), ((T_eL_g)(\eta),\gamma_1,\ldots, \gamma_k)\in
T_gG\times \mathfrak{g}^*\times\stackrel{k}{\cdots}\times\mathfrak{g}^*\,.$

    The momentum map $J\colon (T^1_k)^*G\cong G\times\mathfrak{g}^*\times \stackrel{k}{\ldots}\times
\mathfrak{g}^* \to \mathfrak{g}^*\times \stackrel{k}{\ldots}\times \mathfrak{g}^*$ is given by
    \[
        J(g,\mu_1,\ldots, \mu_k)= Coad^k_g(\mu_1,\ldots, \mu_k)=(Coad_g\mu_1,\ldots, Coad_g\mu_k)\,,
    \]
    for $(g,\mu_1,\ldots, \mu_k)\in G \times \mathfrak{g}^*\times \stackrel{k}{\ldots}\times \mathfrak{g}^*$
 (see Example \ref{KKSth})

    We consider the Hamiltonian function
        \[
            \begin{array}{rrcl}
                H\colon & G\times \mathfrak{g}^*\times \stackrel{k}{\ldots}\times\mathfrak{g}^* &
 \to & \mathbb{R}\\\noalign{\medskip}
                 & (g,\mu_1,\ldots, \mu_k)& \mapsto &
 \displaystyle\frac{1}{2}\displaystyle\sum_{A=1}^k<\mu_A,\mu_A>
            \end{array}
        \]
    where $<\cdot,\cdot>$ denotes the inner product  on $\mathfrak{g}^*$
induced by the inner product on $\mathfrak{g}$. It is trivial that this Hamiltonian is $G$-invariant.

  Throughout this example we consider the isomorphism induced by the inner product
$<\cdot,\cdot>$ given by
    \(\flat_{<\cdot,\cdot>}\colon  \mathfrak{g}  \to  \mathfrak{g}^*\)
    where $(\flat_{<\cdot,\cdot>}(\xi))(\eta)=<\xi,\eta>$ for every $\xi,\eta\in \mathfrak{g}$.

    We consider the $k$-vector field  $(X^H_1,\ldots, X^H_k)$ on
$G\times \mathfrak{g}^*\times\stackrel{k}{\ldots}\times\mathfrak{g}^*$ defined by
        \[
            X^H_A(g,\mu_1,\ldots,\mu_k)=\left( T_eL_g\left(\flat^{-1}_{<\cdot,\cdot>}(\mu_A)\right),
ad^*_{\flat^{-1}_{<\cdot,\cdot>}(\mu_A)}(\mu_1),\ldots, ad^*_{\flat^{-1}_{<\cdot,\cdot>}(\mu_A)}(\mu_k) \right)
        \]
    where $ad^*_\xi\mu\in \mathfrak{g}^*$ is such that $(ad^*_\xi\mu)(\eta)=\mu[\xi,\eta]$.
This $k$-vector field satisfies the following properties:
    \begin{itemize}
        \item Each $X^H_A$ is $G$-invariant.
        \item $X^H_A(g,\mu_1,\ldots,\mu_k)\in \ker\, T_{(g,\mu_1,\ldots,\mu_k)}J,$
for $(g,\mu_1,\ldots, \mu_k)\in G\times \mathfrak{g}^* \times \stackrel{k}{\ldots}\times \mathfrak{g}^*$.

            In fact, if
$(g,\mu_1,\ldots,\mu_k)\in G\times \mathfrak{g}^* \times \stackrel{k}{\ldots}\times \mathfrak{g}^*$
 we have that the transformation $Coad_g$ is a linear isomorphism and thus
            \begin{align*}
            &(T_{(g,\mu_1,\ldots,\mu_k)}J) (X^H_A(g,\mu_1,\ldots,\mu_k))\\ =
& (Coad_g(ad^*_{\flat^{-1}_{<\cdot,\cdot>}(\mu_A)}(\mu_1))- T_e(Coad_g\circ
 Coad_{\mu_1})(\flat^{-1}_{<\cdot,\cdot>}(\mu_A)), \ldots,\\ &
Coad_g(ad^*_{\flat^{-1}_{<\cdot,\cdot>}(\mu_A)}(\mu_k))- T_e(Coad_g\circ
Coad_{\mu_k})(\flat^{-1}_{<\cdot,\cdot>}(\mu_A)))\\=&(0,\ldots, 0)
            \end{align*}
        \item $(X^H_1,\ldots, X^H_k)$ is a solution to the Hamiltonian polysymplectic system, that is,
            \[
                \displaystyle\sum_{A=1}^k\imath_{X^H_A}\omega^A_G=dH\,.
            \]

            Indeed, if $(g,\mu_1,\ldots, \mu_k)\in G\times \mathfrak{g}^* \times
 \stackrel{k}{\ldots}\times \mathfrak{g}^*$ and $(\xi,\nu_1,\ldots,\nu_k)\in
 \mathfrak{g}\times \mathfrak{g}^* \times \stackrel{k}{\ldots}\times \mathfrak{g}^*$, it follows that
            \begin{align*}
            &\left( \sum_{A=1}^k\imath_{X^H_A}\omega^A_G\right)
(g,\mu_1,\ldots, \mu_k)\left((T_eL_g)(\xi),\nu_1,\ldots, \nu_k\right)\\
            =& \ds\sum_{A=1}^k\left(
            \nu_A(\flat^{-1}_{<\cdot,\cdot>}(\mu_A)) - ad^*_{\flat^{-1}_{<\cdot,\cdot>}(\mu_A)}\mu_A (\xi) +
\mu_A[\flat^{-1}_{<\cdot,\cdot>}(\mu_A), \xi]
            \right)
            \\=&dH(g,\mu_1,\ldots, \mu_k)((T_eL_g)(\xi),\nu_1,\ldots, \nu_k)  \ .
            \end{align*}
    \end{itemize}

    We can therefore apply Theorem \ref{reduction_dinamic} and there exists a solution
$(\widehat{X}_1^{H_\mu},\ldots, \widehat{X}_k^{H_\mu})$
to the Hamiltonian polysymplectic system on $J^{-1}(\mu)/G_\mu$
associated with a Hamiltonian function $ H_\mu\colon J^{-1}(\mu)/G_\mu\to \mathbb{R}$
 satisfying that $ H_\mu\circ\pi_\mu = H\circ i$.

     In order to write a solution $(\widehat{X}_1^{H_\mu},\ldots, \widehat{X}_k^{H_\mu})$
 to the reduced Hamiltonian polysymplectic system on $J^{-1}(\mu)/G_\mu$,
we consider the identification between $G$ and $J^{-1}(\mu_1,\ldots, \mu_k)$.
Under this identification, $H\vert_{J^{-1}(\mu_1,\ldots, \mu_k)}$ can be rewritten as follows:
        \[
            \begin{array}{rcl}
                H\vert_{J^{-1}(\mu_1,\ldots,\mu_k)} \colon  G & \to & \mathbb{R}\\\noalign{\medskip}
                 g & \mapsto &\displaystyle\frac{1}{2}\sum_{A=1}^k <Ad^*_g\mu_A, Ad^*_g\mu_A>\,.
            \end{array}
        \]
    Now, applying Theorem \ref{reduction_dinamic} we have
        \begin{equation}\label{reduced_vector_field}
            \widehat{X}^{H_\mu}_A(\nu_1,\ldots, \nu_k)=\left(ad^*_{\flat^{-1}_{<\cdot,\cdot>}\nu_A}\nu_1,\ldots,
 ad^*_{\flat^{-1}_{<\cdot,\cdot>}\nu_A}\nu_k \right)
        \end{equation}
    for each $(\nu_1,\ldots, \nu_k)$ in the $k$-coadjoint orbit $\mathcal{O}_\mu=J^{-1}(\mu)/G_\mu$.
Therefore, $(\widehat{X}^{H_\mu}_1,\ldots, \widehat{X}^{H_\mu}_k)$
is a solution to the reduced Hamiltonian polysymplectic system associated to
the reduced Hamiltonian function given by
        \[
            \begin{array}{rcl}
                H_\mu\colon  \mathcal{O}_\mu\subset \mathfrak{g}^*\times \stackrel{k}{\ldots}
\times\mathfrak{g}^* & \to & \mathbb{R}\\\noalign{\medskip}
                 (\nu_1,\ldots, \nu_k) & \mapsto &\displaystyle\frac{1}{2}\sum_{A=1}^k <\nu_A, \nu_A>  \ .
            \end{array}
        \]

    In the following subsubsection we consider this example in the particular case $G=SO(3)$.

    \subsubsection{Harmonic maps} \cite{CGR-2001,EL-1978}.

        Recall that a smooth map $\varphi\colon M\to N$ between
        Riemannian manifolds $(M,g)$ and $(N,h)$ is \textit{harmonic} if it is a
        critical point of the energy functional $E$, which, when $M$ is
        compact, is defined as
        \[
            E(\varphi)=\int_{M}\frac{1}{2}trace_g\varphi^*h\,dv_g,
        \]
        where $dv_g$ denotes the measure on $M$ induced by its metric and,
        in local coordinates, the expression $\frac{1}{2}trace_g\varphi^*h$
        reads
        \[
            \frac{1}{2} g^{ij}h_{\alpha\beta}\derpar{\varphi^\alpha}{x^i}\derpar{\varphi^\beta}{x^j},
        \]
        $(g^{ij})$ being the inverse of the metric matrix $(g_{ij})$ of $g$ and $(h_{\alpha\beta})$ the metric matrix of $h$.
         (This definition is extended to the case where $M$
         is not compact by requiring the restriction
         of $\varphi$ to every compact domain to be harmonic).

%

        \begin{remark} \rm Some examples of harmonic maps are as follows:
        \begin{itemize}
        \item If $(M,g)=(N,h)$, the identity and the constant map are harmonic.
        \item In the case $k=1$, that is, when $\varphi\colon\r\to N$ is a curve on $N$,
        then $\varphi$
        is a harmonic map if and only if it is a geodesic.
        \item If we consider the case $N=\r$ (with standard metric). Then $\varphi\colon\rk\to\r$ is a
        harmonic map if and only if  it is a harmonic function, that is,
        a solution to the Laplace equation.
        \end{itemize}
        \end{remark}

        In the sequel we consider the case $M=\mathbb{R}^2$ with
        $g_{ij}=\delta_{ij}$ and $N=SO(3)$ with a left-invariant
        metric $h$. Then, we can define a Hamiltonian function
   $$
            \begin{array}{rccl}
                 H\colon & (T^1_2)^*SO(3) & \to & \mathbb{R}\\\noalign{\medskip}
                 & (\alpha^1_g,\alpha^2_g) & \mapsto & \ds\frac{1}{2}
                 \left(\widetilde{h}(\alpha^1_g,\alpha^1_g)+
                 \widetilde{h}(\alpha^2_g,\alpha^2_g)\right) \ ,
            \end{array}
$$
   where $\widetilde{h}$ is the corresponding bundle  metric on $T^*SO(3)$.
        Locally,
        \[
             H(q^i, p^A_i)= \ds\frac{1}{2}h^{ij}p^A_ip^A_j\,.
        \]

        Since $h$ is left-invariant, so is $ H$. Moreover, one may prove, using general results on harmonic
maps (see, for instance \cite{EL-1978}), that if $(X^H_1,X^H_2)$ is a solution to the Hamiltonian
polysymplectic equation associated with $H$ and $\gamma\colon \mathbb{R}^2\to SO(3)$
 is an integral submanifold of the distribution generated by $X^H_1$ and $X^H_2$, then $\gamma$
 is a harmonic map.

        On the other hand, as we have seen in the general situation referred to
        the previous example \ref{general group example}, we have that under the assumptions
        of Theorem \ref{MW poly th}, there exist $(Y_1,Y_2)$
        a solution to the Hamiltonian polysymplectic system on $J^{-1}(\mu)/G_{\mu}$
        associated with a Hamiltonian function
        $ H_\mu\colon J^{-1}(\mu)/G_\mu=\mathcal{O}_\mu\to \mathbb{R}$
        satisfying $ H_\mu\circ \pi_\mu =  H\circ i$; that is,
        $(Y_1,Y_2)$ is a solution to the equation
        \[
            \imath_{Y_1}\omega^1_\mu + \imath_{Y_2}\omega^2_\mu =d H_\mu\,.
        \]
         In this particular case, the expression of the polysymplectic forms
         $\omega^1_\mu,\omega^2_\mu$ is described in Section \ref{KKSth} (see (\ref{kkspolystruc})).

   In accordance with the results and identifications in Section \ref{KKSth}, we consider the following cases:
         \begin{enumerate}
            \item $\pi^0_1$ and $\pi^0_2$ are linearly dependent and $(\pi^0_1,\pi^0_2)\neq 0$.
Assume that $\pi^0_1\neq 0$ and $\pi^0_2=\lambda_0\pi^0_1$ with $\lambda_0\neq 0$.
In this case $\mathcal{O}_{(\pi^0_1,\pi^0_2)}=S^2(||\pi^0_1||)$ and
                \[
                    T_{(\pi,\lambda_0\pi)}\mathcal{O}_{(\pi^0_1,\pi^0_2)}=\{
(\mathbf{v},\lambda_0\mathbf{v})\in\mathbb{R}^3\times \mathbb{R}^3 | \, \mathbf{v}\in T_\pi S^2(||\pi^0_1||)
                    \}\,.
                \]

                On the other hand,
                \[
                    T_\pi S^2(||\pi^0_1||)=\{ \xi_{\mathbb{R}^3}(\pi) / \xi\in \mathfrak{so}(3)\equiv\mathbb{R}^3\}
                \]
                and $\xi_\mathbb{R^3}(\pi)=\xi\times \pi$, for every $\xi\in \mathfrak{so}(3)\equiv\mathbb{R}^3$.

                Therefore, at a point $(\pi,\lambda_0\pi)\in \mathcal{O}_{(\pi^0_1,\pi^0_2)}=S^2(||\pi^0_1||)$,
 the solution to the reduced Hamiltonian polysymplectic system is (see \ref{reduced_vector_field})
                \begin{align*}
                    \widehat{X}^{H_\mu}_1(\pi,\lambda_0\pi)
     =& \big(ad^*_{\flat^{-1}_{<\cdot,\cdot>}(\pi)}\pi, ad^*_{\flat^{-1}_{<\cdot,\cdot>}(\pi)}(\lambda_0\pi) \big)
 = \big(\pi\times \flat_{<\cdot,\cdot>^{-1}}(\pi), \lambda_0\pi\times \flat_{<\cdot,\cdot>^{-1}}(\pi) \big)\\
\noalign{\medskip}
                    \widehat{X}^{H_\mu}_2(\pi,\lambda_0\pi)
                    =& \big(ad^*_{\flat^{-1}_{<\cdot,\cdot>}(\lambda_0\pi)}\pi, ad^*_{\flat^{-1}_{<\cdot,\cdot>}
(\lambda_0\pi)}(\lambda_0\pi) \big)
                    = \big(\lambda_0\pi\times \flat_{<\cdot,\cdot>^{-1}}(\pi), \lambda_0^2\pi\times \flat_{<\cdot,
\cdot>^{-1}}(\pi) \big) \\\noalign{\medskip}
                    =& \lambda_0 \widehat{X}^{H_\mu}_1(\pi,\lambda_0\pi) \ .
                \end{align*}
            \item $\pi^0_1$ and $\pi^0_2$ are linearly independent. In this case, we know that there exists a
diffeomorphism between $\mathcal{O}_{(\pi^0_1,\pi^0_2)}$ and $SO(3)$, where
                \[
                    \mathcal{O}_{(\pi^0_1,\pi^0_2)}=\{ (A\pi^0_1,A\pi^0_2) \, | \, A\in SO(3) \}  \ .
                \]
                Therefore, from (\ref{reduced_vector_field}), we have that
$(\widehat{X}^{H_\mu}_1,\widehat{X}^{H_\mu}_2)$
is a solution to the reduced Hamiltonian polysymplectic system where
                \begin{align*}
                    \widehat{X}^{H_\mu}_1(A\pi^0_1,A\pi^0_2)
=& \big(ad^*_{\flat^{-1}_{<\cdot,\cdot>}(A\pi^0_1)}(A\pi^0_1),
ad^*_{\flat^{-1}_{<\cdot,\cdot>}(A\pi^0_1)}(A\pi^0_2) \big)
                    \\\noalign{\medskip}
                    = &\big((A\pi^0_1)\times \flat_{<\cdot,\cdot>^{-1}}(A\pi^0_1), (A\pi^0_2)
\times \flat_{<\cdot,\cdot>^{-1}}(A\pi^0_1) \big)\\\noalign{\medskip}
                    \widehat{X}^{H_\mu}_2(A\pi^0_1,A\pi^0_2)
                    =& \big(ad^*_{\flat^{-1}_{<\cdot,\cdot>}(A\pi^0_2)}(A\pi^0_1),
 ad^*_{\flat^{-1}_{<\cdot,\cdot>}(A\pi^0_2)}(A\pi^0_2) \big)
                    \\\noalign{\medskip}
                    = & \big((A\pi^0_1)\times \flat_{<\cdot,\cdot>^{-1}}(A\pi^0_2), (A\pi^0_2)\times
 \flat_{<\cdot,\cdot>^{-1}}(A\pi^0_2) \big)\ .
                \end{align*}
         \end{enumerate}

\section{Conclusions and future work}

We study the reduction of polysymplectic manifolds
and Hamiltonian polysymplectic systems, such as those
that appear in some types of classical field theories.

First, we have given an example that shows a mistake
in the reduction scheme proposed by G\"{u}nther.

Then,
after stating the guidelines for reduction of a polysymplectic manifold
by a generic submanifold, we prove a generalized version of the Marsden-Weinstein reduction theorem for a
polysymplectic manifold $M$ in the presence of an equivariant momentum map for a polysymplectic action
on $M$. In this paper, we give the conditions for the polysymplectic reduction.  In fact, a new   additional hypothesis
must be added to the usual ones (regular values of the momentum map,
free and proper actions);
namely, the constancy of the rank of the characteristic foliation on the level set of the momentum map
 corresponding to a fixed value $\mu\in \mathfrak{g}^*$, and the fact that the leaves of this foliation are the
 orbits of the action of the isotropy group $G_\mu$ on the level set.
One of the main goals of this work is to study what
conditions ensure that this hypothesis holds
(see Section \ref{MWPR}).
Assuming all these conditions, we prove that the quotient space
is a manifold that inherits a polysymplectic structure
from the initial one.
In this way, the limitations of the reduction theorem presented in \cite{MRS-2004},
which are referred in the introduction, are overcome and corrected.

As an application of our theorem, we analyze the case of the product of symplectic manifolds and the
particular case of reduction of the standard model of polysymplectic ($k$-symplectic) manifold:
the cotangent bundle of $k^1$-covelocities.
Furthermore, we generalize the
{\sl Kirillov-Kostant-Souriau theorem} to the case of polysymplectic manifolds.

Finally, the reduction of polysymplectic Hamiltonian systems
is also studied as a natural continuation of the previous results,
showing how under the same hypothesis as above,
and assuming the invariance of the Hamiltonian function,
a new Hamiltonian polysymplectic system is defined in the quotient space.
These results are applied to analyzing the problem of reduction
of Hamiltonian polysymplectic systems defined in
cotangent bundles of $k^1$-covelocities, which admit a suitable decomposition and,
as a particular case,
the harmonic maps.

Another (possible) potential application of the previous results could be the following one. Some physical theories admit a Lagrangian formulation, as a classical field theory of first order, with a $G$-invariant regular Lagrangian function $L$ which is defined in the tangent bundle of $k$-velocities $T^1_kQ$ associated with a manifold $Q$ (for example, this situation appears when one deals with the dynamics of molecular strands; in fact, for this theory, $k=2$, $G$ is the special orthogonal group $SO(3)$ and $Q$ is the semi-direct product $SE(3)=SO(3)\circledS\mathbb{R}^3$ of $SO(3)$ with the abelian Lie group $\mathbb{R}^3$, see \cite{EGHPR-2009}). So, one may obtain the corresponding Lagrange-Poincar\'{e} field equations on the reduced space $T^1_kQ/G$ (see \cite{EGHPR-2009} for the particular case of molecular strands).

Using that the Lagrangian function $L$ is regular, one may develop a Hamiltonian formulation of the theory with $G$-invariant Hamiltonian function $H$ which is defined in the cotangent bundle of $k^1$-covelocities $(T^1_k)^* Q$. Moreover, the solutions of the Hamiltonian polysymplectic equation for $H$ are solutions of the Hamilton-De Donder-Weyl equations for the corresponding Hamiltonian classical field theory.

Thus, for a solution of the first equations which satisfies the hypotheses of Theorem \ref{reduction_dinamic}, one could obtain a solution of the reduced Hamiltonian polysymplectic system on the corresponding reduced space $J^{-1}(\mu)/G_{\mu}$.

It would be interesting to relate these solutions with the solutions of the Lagrange-Poincar\'{e} field equations on $T^1_kQ/G$. Note that the space of orbits $(T^1_k)^*Q/G$ admits a poly-Poisson structure (see \cite{IMV-2012}) and it seems likely that the reduced spaces $J^{-1}(\mu)/G_\mu$ can be leaves of the canonical polysymplectic foliation in $(T^1_k)^*Q/G$ (for the definition of the canonical polysymplectic foliation associated with a polysymplectic structure, see \cite{IMV-2012}). Then, the Legendre transformation between $T^1_kQ/G$ and $(T^1_k)^*Q/G$, induced by the reduced Lagrangian function on $T^1_kQ/G$, should relate the solutions of both equations.

Anyway, this paper is the first step towards a more ambitious program of reduction
(``a la Marsden-Weinstein'') of geometric classical field theories.
In particular, since the multisymplectic formulation constitutes the most general
geometric framework for describing classical field theories,
our next objective is to extend the results obtained here to multisymplectic manifolds,
in such a way that they can be applied to reduce multisymplectic Hamiltonian systems.

\appendix
\section{Examples of polysymplectic manifolds}
\label{examplepoly}

    In this appendix we describe some typical examples of polysymplectic manifolds.

    \subsection{The product of symplectic manifolds}\label{productsymp}

    Let $M_A$ be a symplectic manifold with symplectic form $\tilde{\omega}^A$, for $A\in \{1,\ldots, k\}$.

    We consider the product manifold
    \[
        M= M_1\times \cdots \times M_k
    \]
    and the $2$-form $\omega^A$ on $M$ given by
    \[
        \omega^A=(pr_A)^*(\tilde{\omega}^A)\,,
    \]
    where $pr_A\colon M\to M_A$ is the canonical projection, for $A\in \{1,\ldots, k\}$.

    Then, it is clear that $(\omega^1,\ldots, \omega^k)$ is a $k$-polysymplectic structure on $M$.
    \subsection{The cotangent bundle of $k^1$-covelocities of a manifold}\label{t1k*q}

        Let $Q$ be a differentiable manifold, $\dim Q = n$, and $\pi_Q: T^{\;*}Q \to Q$ its cotangent bundle.
        Denote by $(T^1_k)^*Q$  the Whitney sum $T^{\;*}Q \oplus \stackrel{k}{\dots}
        \oplus T^{\;*}Q$ of $k$ copies of $T^{\;*}Q$, with projection $\pi^k_Q\colon (T^1_k)^*Q\to Q$.

        $(T^1_k)^*Q$ can be identified with the manifold $J^1(Q,\rk)_0$ of 1-jets of maps
 $\sigma\colon Q\to\rk$ with target at $0\in \rk$,
        the diffeomorphism is given by
        \[
            \begin{array}{ccc}
                J^1(Q,\r^k)_0 & \equiv & T^{\;*}Q \oplus \stackrel{k}{\dots} \oplus T^{\;*}Q \\
                j^1_{q,0}\sigma  & \equiv & (d\sigma^1(q), \dots ,d\sigma^k(q))\ ,
            \end{array}
        \]
        where $\sigma^A= \pi^A \circ \sigma:Q \longrightarrow \r$ is the $A^{th}$ component of $\sigma$,
        and  $\pi^A:\r^k \to \r$ is the canonical projection onto the $A^{th}$ component, for $ A = 1, \ldots ,
        k$. $(T^1_k)^*Q$ is called {\sl the cotangent bundle of $k^1$-covelocities of the manifold $Q$}.

        If $(q^i)$ are local coordinates on $U \subseteq Q$, then the induced local coordinates  $(q^i ,
        p^A_i)$ on $(\pi^k_Q)^{-1}(U)=(T^1_k)^{\;*}U$ are given by
        \[
            q^i(\alpha^1_q, \ldots ,\alpha^k_q) = q^i(q)\, , \quad p^A_i(\alpha^1_q, \ldots ,\alpha^k_q) =
\alpha^A_q\left(\ds\frac{\partial}{\partial q^i}\Big\vert_q \right)\, ,\qquad 1\leq i\leq n;\,1\leq A\leq k\,.
        \]

        On $ (T^1_k)^*Q$, we consider the differential forms
        \[
            \theta^A= (\pi_Q^{k,A})^*\theta\, , \quad \omega^A= (\pi_Q^{k,A})^*\omega\, ,
        \]
        where  \(\omega=-d\theta=dq^i \wedge dp_i\) is the canonical symplectic form on $T^*Q$,
        $\theta=p_i \, dq^i$ is the Liouville $1$-form on $T^*Q$ and
        $\pi_Q^{k,A}:  (T^1_k)^*Q \rightarrow T^*Q$ is the projection defined by
        \[
            \pi_Q^{k,A}(\alpha^1_q, \ldots ,\alpha^k_q)=\alpha^A_q\,           .
        \]
         Obviously, $\omega^A = -d\theta^A$.

        In local natural coordinates, we have
        \begin{equation}\label{locexp}
            \theta^A = \displaystyle \, p^A_i\,dq^i \, ,  \quad \omega^A = \displaystyle  dq^i \wedge dp^A_i\, .
        \end{equation}

        A simple inspection of their expressions in local coordinates shows that the forms
        $\omega^A$ are closed and the relation (\ref{poly-cond}) holds; that is, $(\omega^1,\ldots, \omega^k)$
        is a $k$-polysymplectic structure on $(T^1_k)^* Q$.

    \subsection{Frame bundle}

        Let $LM$ be the frame bundle of $M$; that is, the manifold of all the vector space bases
         in all the tangent spaces at
         the various points of $M$. This bundle is a special type of principal bundle in the sense that its
geometry is fundamentally
          tied to the geometry of $M$. This relation can be expressed by means
           of the vector-valued $1$-form $\vartheta=\ds\sum_{A=1}^k\vartheta^A r_A\in \Omega^1(LM,\r^n)$
called the \textit{solder form}.
            This form is defined by
            \[
                \begin{array}{rccl}
                    \vartheta(u)\colon & T_u(LM) & \to & \r^n\\\noalign{\medskip} & X_u & \mapsto & \vartheta(u)
(X_u)=u^{-1}(T_u\pi (X_u)) \, ,
                \end{array}
            \]
        where $\pi\colon LM\to M$ is the canonical projection and
        $  u\colon \r^n\to T_xM$ a point  of $LM$.

        The solder form endows $LM$ with a $n$-polysymplectic structure given by
            \[
                \omega^A=d\vartheta^A,\quad A=1,\ldots, n\,.
            \]
(See \cite{Norris} for more details).

%
%
%
%

    \subsection{$k$-coadjoint orbits}\label{k-coadjoint orbit}

        Before describing this new example of a polysymplectic manifold,
it is necessary to recall the symplectic structure of the coadjoint orbit of a Lie group
(for more details see \cite{AM-1978}, page 303).

        Let $G$ be a Lie group, $\mathfrak{g}$ its Lie algebra. We consider the  \textit{coadjoint action}
            \[
                \begin{array}{lccl}
                    Coad\colon & G\times \mathfrak{g}^* &\to & \mathfrak{g}^*\\\noalign{\medskip}
                            & (g,\mu) & \mapsto & Coad(g,\mu)=\mu\circ Ad_{g^{-1}}
                \end{array}
            \]
        and the orbit of $\mu \in\mathfrak{g}^*$ in $\mathfrak{g}^*$ under this action,
            \[
                \mathcal{O}_\mu=\{Coad(g,\mu) \ \mid \  g\in G\}\,.
            \]

        It is well known that \(\mathcal{O}_\mu\) has a symplectic structure $\omega_{\mu}$
defined by the expression
            \begin{equation}\label{symp-orbit}
                \omega_{\mu}(\nu)\left(\xi_{\mathfrak{g}^*}(\nu),\eta_{\mathfrak{g}^*}(\nu)\right)=-\nu [\xi,\eta]
            \end{equation}
        where $\nu$ is an arbitrary point of $\mathcal{O}_{\mu}$, $ \xi_{\mathfrak{g}^*}(\nu),
\eta_{\mathfrak{g}^*}(\nu)\in T_\nu\mathcal{O}_\mu$.

        Let $(\mu_1,\ldots, \mu_k)$ be an element of
$\mathfrak{g}^*\times\stackrel{k}{\ldots}\times\mathfrak{g}^*$.
        We define the \textit{$k$-coadjoint orbit} as the orbit of
        $(\mu_1,\ldots, \mu_k)$ in $\mathfrak{g}^*\times\stackrel{k}{\ldots}\times\mathfrak{g}^*$, that is,
            \[
                \mathcal{O}_{(\mu_1,\ldots, \mu_k)} =\{Coad^k(g,\mu_1,\ldots, \mu_k)\ \mid \  g\in G\} \,,
            \]
     $Coad^k$  being the $k$-coadjoint action defined in (\ref{coad^k}). The space
$\mathcal{O}_{\mu_1,\ldots, \mu_k}$ was considered in \cite{Gunther-1987}.
 In fact, in \cite{Gunther-1987}, $\mathcal{O}_{\mu_1,\ldots, \mu_k}$
was called the polycoadjoint orbit by $(\mu_1,\ldots, \mu_k)$.

        Next, we will recall the definition of the $k$-polysymplectic structure on
 $\mathcal{O}_{\mu_1,\ldots, \mu_k}$ which was introduced in \cite{Gunther-1987}.

        \begin{lemma}\label{lem-orbit-1}
            For every $(\nu_1,\ldots,\nu_k)\in \mathcal{O}_{(\mu_1,\ldots, \mu_k)}$ we have that
                \[
                    T_{(\nu_1,\ldots,\nu_k)}\mathcal{O}_{(\mu_1,\ldots, \mu_k)}=
                    \{\xi_{\mathfrak{g}^*\times\stackrel{k}{\ldots}\times\mathfrak{g}^*}(\nu_1,\ldots,\nu_k)
 \ \mid \  \xi\in \mathfrak{g}\} \ ,
                \]
            where $\xi_{\mathfrak{g}^*\times\stackrel{k}{\ldots}\times\mathfrak{g}^*}$
            is the \textrm{infinitesimal generator} of the $k$-coadjoint action  corresponding to $\xi$.
        \end{lemma}

        \proof This is a well-known result (see for example \cite{AM-1978} p. 267).
        \qed

        \begin{lemma}\label{lem-orbit-2}
            For every $A=1,\dots, k$ and each $(\nu_1,\ldots,\nu_k)\in \mathcal{O}_{(\mu_1,\ldots, \mu_k)}$
 we obtain that
                \[
                    (pr_A)_*(\nu_1,\ldots,\nu_k)\left(\xi_{\mathfrak{g}^*\times\stackrel{k}{\ldots}
\times\mathfrak{g}^*}(\nu_1,\ldots,\nu_k)
                    \right)=\xi_{\mathfrak{g}^*}(\nu_A) \ ,
                \]
            where $pr_A$ is the canonical projection
                \[
                   \begin{array}{lccc}
                        pr_A\colon & \mathcal{O}_{(\mu_1,\ldots, \mu_k)} & \to & \mathcal{O}_{\mu_A}\\
\noalign{\medskip}
                         & (\nu_1,\ldots,\nu_k) & \mapsto & \nu_A \ .
                    \end{array}
                \]
        \end{lemma}
        \proof
        As the relation $pr_A\circ Coad^k_{(\nu_1,\ldots,\nu_k)} = Coad_{\nu_A}$ holds, we obtain
            \[
                (pr_A)_*(\nu_1,\ldots,\nu_k)\left(\xi_{\mathfrak{g}^*\times\stackrel{k}{\ldots}\times
\mathfrak{g}^*}(\nu_1,\ldots,\nu_k)
                    \right) = T_e(pr_A\circ Coad^k_{(\nu_1,\ldots,\nu_k)})(\xi)
                    =T_eCoad_{\nu_A}(\xi)=\xi_{\mathfrak{g}^*}(\nu_A)\,.
            \]
        \qed

        As a consequence of the above lemma we can consider the following relations:
        \begin{equation}\label{ident}
            \begin{array}{ccc}
                T_{(\nu_1,\ldots,\nu_k)}\mathcal{O}_{(\mu_1,\ldots, \mu_k)}
                & \subseteq & T_{\nu_1}\mathcal{O}_{\mu_1}\times \ldots \times T_{\nu_k}\mathcal{O}_{\mu_k}\\
\noalign{\medskip}
                \xi_{\mathfrak{g}^*\times\stackrel{k}{\ldots}\times\mathfrak{g}^*}(\nu_1,\ldots,\nu_k)
                & \equiv & \left(\xi_{\mathfrak{g}^*}(\nu_1),\ldots, \xi_{\mathfrak{g}^*}(\nu_k)\right).
            \end{array}
        \end{equation}

        \begin{prop}\label{A4}
            Let $\omega_{\mu_A}$ be the symplectic structure of the coadjoint orbit
             $\mathcal{O}_{\mu_A}$ at $\mu_A$, then the family $(\omega_\mu^{1},\ldots, \omega_\mu^k)$
 given by
                \[
                    \omega_\mu^A\colon = (pr_A)^*\omega_{\mu_A}
                \]
            is a $k$-polysymplectic structure on the $k$-coadjoint orbit $\mathcal{O}_{(\mu_1,\ldots, \mu_k)}$
 at $\mu=(\mu_1,\ldots, \mu_k)$.
        \end{prop}

        \proof By definition, every $\omega_\mu^A$ is a closed $2$-form on
$\mathcal{O}_{(\mu_1,\ldots, \mu_k)}$.
        Now we have to prove that $\displaystyle\bigcap_{A=1}^k\ker\,\omega_\mu^A=0$.

        From Lemmas \ref{lem-orbit-1} and \ref{lem-orbit-2} and the expression (\ref{symp-orbit})
        of the symplectic form $\omega_{\mu_A}$,
              if $(\nu_1,\ldots,\nu_k)$ is an arbitrary point of
$\mathfrak{g}^*\times\stackrel{k}{\ldots}\times\mathfrak{g}^*$,
        we obtain that
            \begin{equation}\label{polysymp-orbit}
                \begin{array}{lcl}
                    \omega_\mu^A (\nu_1,\ldots,\nu_k)\left(\xi_{\mathfrak{g}^*\times\stackrel{k}{\ldots}\times
\mathfrak{g}^*}(\nu_1,\ldots,\nu_k) , \eta_{\mathfrak{g}^*\times\stackrel{k}{\ldots}\times\mathfrak{g}^*}
(\nu_1,\ldots,\nu_k) \right)                    &=&
                    \\\noalign{\medskip}
                    \left[(pr_A)^*\omega_{\mu_A}\right]\left(\xi_{\mathfrak{g}^*\times\stackrel{k}{\ldots}\times
\mathfrak{g}^*}(\nu_1,\ldots,\nu_k) , \eta_{\mathfrak{g}^*\times\stackrel{k}{\ldots}\times\mathfrak{g}^*}
(\nu_1,\ldots,\nu_k) \right)                    &=&
                    \\\noalign{\medskip}
                    \omega_{\mu_A}(\nu_A)\left(\xi_{\mathfrak{g}^*}(\nu_A),
                    \eta_{\mathfrak{g}^*}(\nu_A)\right) =-\nu_A[\xi,\eta]\ .
                \end{array}
            \end{equation}

        Let $\xi_{\mathfrak{g}^*\times\stackrel{k}{\ldots}\times\mathfrak{g}^*}(\nu_1,\ldots,\nu_k)$
        be an element of $\displaystyle\bigcap_{A=1}^k\ker\,\omega_\mu^A$.
        As a consequence of (\ref{polysymp-orbit}), we obtain that $\nu_A[\xi,\eta]=0$, for every $\eta\in\mathfrak{g}$, and
        this is equivalent to $\xi_{\mathfrak{g}^*}(\nu_A)=0$.
        Therefore, using the identification (\ref{ident}), we obtain that
        $\xi_{\mathfrak{g}^*\times\stackrel{k}{\ldots}\times\mathfrak{g}^*}(\nu_1,\ldots,\nu_k)=0$
        and thus $\displaystyle\bigcap_{A=1}^k\ker\,\omega_\mu^A=0$. \qed

\section*{Acknowledgments}
We acknowledge the partial financial support of the {\sl
Ministerio de Ciencia e Innovaci\'on} (Spain), projects
MTM2009-13383, MTM2011-22585, MTM2012-34478, and  MTM2011-15725-E,
of the {\sl Canary Government}, Project ProdID20100210,
and of {\sl Gobierno de Arag\' on} project E24/1.

We wish
to thank to Mr. Jeff Palmer for his assistance in preparing the
English version of the manuscript.

\end{document}